\documentclass[sn-mathphys,Numbered]{sn-jnl}

\usepackage{graphicx}%
\usepackage{multirow}%
\usepackage{amsmath,amssymb,amsfonts}%
\usepackage{amsthm}%
\usepackage{mathrsfs}%
\usepackage[title]{appendix}%
\usepackage{xcolor}%
\usepackage{textcomp}%
\usepackage{manyfoot}%
\usepackage{booktabs}%
\usepackage{algorithm}%
\usepackage{algorithmicx}%
\usepackage{algpseudocode}%
\usepackage{listings}%
\usepackage{natbib}
\usepackage{adjustbox}

\theoremstyle{thmstyleone}%
\newtheorem{theorem}{Theorem}%
\newtheorem{corollary}{Corollary}[theorem]
\newtheorem{lemma}{Lemma}

\newtheorem{result}{Result}
\theoremstyle{thmstyletwo}%

\theoremstyle{thmstylethree}%
\begin{document}

\title[Remote state preparation by multiple observers using a single copy of a two-qubit entangled state]{Remote state preparation by multiple observers using a single copy of a two-qubit entangled state}


\author*[1]{\fnm{Shounak} \sur{Datta}}\email{shounak.datta@bose.res.in; shounaknew@gmail.com}

\author[2,3,4,5]{\fnm{Shiladitya} \sur{Mal}}\email{shiladitya.27@gmail.com}
\equalcont{These authors contributed equally to this work.}

\author[2]{\fnm{Arun K.} \sur{Pati}}\email{akpati@hri.res.in}
\equalcont{These authors contributed equally to this work.}

\author[1]{\fnm{A. S.} \sur{Majumdar}}\email{archan@bose.res.in}
\equalcont{These authors contributed equally to this work.}

\affil*[1]{\orgname{S. N. Bose National Centre for Basic Sciences}, \orgaddress{\street{Block JD, Sector III, Salt Lake}, \city{Kolkata}, \postcode{700 106}, \country{India}}}

\affil[2]{\orgname{Harish-Chandra Research Institute, HBNI}, \orgaddress{\street{Chhatnag Road, Jhunsi}, \city{Allahabad}, \postcode{211 019}, \country{India}}}

\affil[3]{\orgdiv{Physics Division}, \orgname{National Center for Theoretical Sciences}, \orgaddress{\city{Taipei}, \postcode{10617}, \country{Taiwan}}}

\affil[4]{\orgdiv{Department of Physics and Center for Quantum Frontiers of Research and Technology (QFort)}, \orgname{National Cheng Kung University}, \orgaddress{\city{Tainan}, \postcode{701}, \country{Taiwan}}}

\affil[5]{\orgdiv{Centre for Quantum Science and Technology}, \orgname{Chennai Institute of Technology}, \orgaddress{\city{Chennai}, \postcode{600069}, \country{India}}}


\abstract{We consider a scenario of remote state preparation (RSP) of qubits in the context of sequential network scenario. A single copy of an entangled state is shared between Alice on one side, and several Bobs on the other, who sequentially perform unsharp single-particle measurements in order to prepare a specific state. In the given scenario without any shared randomness between the various Bobs, we first determine the classical bound of fidelity for the preparation of remote states by the Bobs. We then show that there can be at most 6 number of Bobs who can sequentially and independently prepare the remote qubit state in Alice's lab with fidelity exceeding the classical bound in the presence of shared quantum correlations. The upper bound is achieved when the singlet state is initially shared between Alice and the first Bob and every Bob prepares a state chosen from the equatorial circle of the Bloch sphere. Then we introduce a new RSP protocol for non-equatorial ensemble of states. The maximum number of Bobs starts to decrease from six when either the choice of remote states is shifted from the equatorial circle towards the poles of the Bloch sphere, or when the  initial state shifts towards non-maximally entangled pure and mixed states.}




\maketitle

\section{Introduction}
A central issue in quantum communication and computation is to identify the tasks that can be performed with quantum resources. One of the most important tasks is to send quantum information or prepare some quantum states at remote locations. Teleportation~\cite{Bennett} is one such task where an unknown quantum state is transferred using a quantum channel without physically transporting the system. For teleporting a two-level system, or qubit, it is known that the protocol requires Bell-state measurement~\cite{BSM}, 1 ebit of entanglement and 2 cbits of classical communication from the sender to the receiver. It was shown later~\cite{Cerf} that if the requirement of entanglement is relaxed, then using shared local hidden variables, any projective measurement on an unknown state can be simulated at a remote location, dubbed as classical teleportation, which can be performed with the aid of $2.19$ cbits on average.

The advantage of the quantum realm was further demonstrated by the protocol of Remote State Preparation (RSP)~\cite{Pati,Bennett1}, a variant of teleportation, where the state to be prepared is known to the sender but oblivious to the receiver. Here, only 1 cbit is necessary along with the requirement of 1 ebit. The task of remote state preparation has been generalised in several ways in the literature~\cite{Oblivious,Optimal}, and realised in magnetic~\cite{NMR} as well as optical systems by means of using single mode photonic qubit~\cite{SMPQ}, polarised photons~\cite{Polar} via spontaneous parametric down conversion, and decoherent channels~\cite{DC}. RSP has found several applications in quantum information such as deterministic creation of single-photon states~\cite{Christ}, preparation of single-photon hybrid entanglement~\cite{Jiao}, initializing atomic quantum memory~\cite{Rosenfeld}, and preparing qubits in quantum nodes~\cite{Nguyen}. Similar to the spin systems, two-component Bose-Einstein condensates~\cite{Manish} and two degrees of freedom single photon beams~\cite{Wang,Cameron} can  be remotely prepared in different optical arrangements.  Being a quantum communication task like teleportation, RSP can be employed as a subroutine in a bigger network, {\it e.g.}, transferring information between various nodes of a quantum processor.

In remote state preparation, it is possible to have a trade-off between classical communication and shared entanglement to a extent~\cite{Bennett1}. Depending on the choice of the input ensemble from where the states are to be prepared, and how much information of the state is to be revealed to the receiver, there are different protocols of RSP~\cite{Oblivious, Optimal, DC, Jiao, Rosenfeld, Paris}. Instead of the most general encoding by the sender and the decoding operation by the receiver, various naturally restricted and experiment friendly protocols of RSP have been investigated~\cite{Optimal, Nguyen, Discord}. Considering such an encoding-decoding protocol and input ensemble taken as an equatorial circle on Bloch sphere, it has been shown that quantum correlation beyond entanglement (QCBE) can be helpful for RSP~\cite{Discord}. Non-vanishing geometric discord~\cite{GD}, a measure of QCBE, implies non-zero fidelity of RSP and more interestingly, there are some separable states which can outperform entangled state as well~\cite{Discord}. In an another work, simultaneous correlation between mutually unbiased bases has been shown to be responsible for a specific RSP protocol~\cite{Som}. It may be noted that if one considers the most general encoding-decoding operations,  separable states can't outperform entangled states in RSP for equatorial input states~\cite{Horodecki14}. In order to identify genuine quantum advantage in RSP, benchmark for classical resources has to be found for a specific protocol of RSP~\cite{Horodecki14}. The  non-classical nature of RSP is further revealed in terms of steerable features of dynamical processes~\cite{Chen}.

In the present work, we explore the performance of  remote state preparation  in the context of a sequential measurement scenario. We investigate the possibility of implementing RSP at Alice's lab by multiple distant Bobs who act sequentially and independently of each other.  Alice owns half of a single copy of an entangled state while the other half is shared sequentially by several Bobs. Each Bob individually wants to convey a message to Alice by preparing a particular state at Alice's lab up to some prefixed level of tolerance. Any pure state of a qubit is represented by a point on the Bloch sphere, which is characterised by two parameters, $\theta$ and $\phi$. Before starting the protocol, the value of $\theta$ is agreed in advance between Alice and all Bobs, i.e. the circle from which the state would be prepared. Based on the communication made by every Bob, Alice applies a suitable unitary to obtain the desired state. At any instant, a particular Bob remains ignorant about the encoding activities of the previous Bobs. Now,  the pertinent question is: how many Bobs can reliably convey the message to Alice within some error tolerance?

In an ideal scenario employing projective measurements, the sender can, in principle, deterministically prepare the desired remote qubit~\cite{Hayashi} at the receiver's end with the expense of complete breakdown of the entanglement. However, since we want RSP to be accomplished by multiple senders in the present case, here all the Bobs (except the last one) have to measure weakly so that some amount of quantum correlation is left to be shared by the subsequent Bobs and Alice ~\cite{Silva, Mal}. 

The novel framework of sequential sharing of resource states has been employed earlier in the context of Bell-nonlocality~\cite{Silva, Mal} and has been experimentally demonstrated too~\cite{Hu,Schiavon,Feng}. Studies on sequential sharing of quantum correlations have been extended  in several different directions, such as Bell-type nonlocality with multiple settings~\cite{Das,Cabello}, detection of entanglement~\cite{Bera,Mdgrm,Srivastava}, detecting bipartite~\cite{Sasmal, Brunner}, and tripartite steerability~\cite{Gupta}, steerability of local quantum coherence~\cite{Datta}, potential shareability of a teleportation channel~\cite{Roy}. Fruitfulness of the sequential scenario has been demonstrated by showing how unbounded randomness can be extracted from a two-qubit resource state~\cite{RNG}. Moreover, the resource theoretic advantage of utilizing a single copy of a two-qubit entangled state has been demonstrated in context of a sequential network as compared to using multiple entangled copies in a non-sequential scenario~\cite{Ddmhm}. In another interesting variant of the sequential protocol, it is shown that unbounded number of Bobs can exhibit Bell nonlocality with a single Alice~\cite{Colbeck}.

In our work we consider a specific RSP protocol which interestingly exhibits quantum advantage even for mixed separable shared states with quantum correlations~\cite{Discord}, when the states to be remotely prepared are chosen from an equatorial circle on the Bloch sphere. We further introduce a new RSP protocol with non-equatorial states as input ensemble, {\it i.e.}, states taken form a circles with varying polar angles,  and also in the context of the sequential scenario. In all these cases we first compute the classical limit of fidelity to prepare an unknown remote state without using any quantum resource, where 1 cbit of information is allowed to be transferred from the sender (Bob) to the receiver (Alice) through a classical channel. This is calculated by taking the average of input states over a particular circle on the Bloch sphere and turns out to be different from $\frac{2}{3}$ as found to be optimal in the standard scheme of quantum teleportation~\cite{Massar}. Interestingly, it is found that the maximum number of observers who can successfully realise the task of RSP with non-classical fidelity, depends upon the choice of the plane of the Bloch sphere. Considering a singlet state shared initially, we find that six Bobs become successful in this task if the state to be prepared is selected from the equator of the Bloch sphere. If the chosen remote state shifts from the equatorial circle towards the poles of the Bloch sphere, the maximum number of Bobs demonstrating the quantum advantage of RSP reduces, reaching zero at the two poles.  We further find that the upper bound on the number of senders (Bobs) becomes lower than six when the initial state is a non-maximally or mixed entangled state. 

We organize the paper as follows. In Sec.\ref{R2}, we provide first a brief overview of RSP and its optimal fidelity under a classical strategy. Next we present in Sec.\ref{R2.3} the definition and function of our RSP-protocol in the framework of multiple observers using a single copy of a two-qubit entangled state. In Sec.\ref{R3}, we perform a systematic investigation on the upper bound of observers for the manifestation of RSP. Concluding  remarks are presented in Sec.\ref{R4}.

\section{Backdrop} \label{R2}

\subsection{Brief introduction to Remote State Preparation}

To explain the quantum strategy for the remote state preparation in bipartite qubit scenario, we consider that two spin-$\frac{1}{2}$ particles are prepared in a singlet state, $|\psi^-\rangle = \frac{1}{\sqrt{2}} (|01\rangle -|10\rangle)$ where the first particle is possessed by the sender, let us call Bob and the second particle is possessed by the receiver, let us call Alice. Though $\lbrace |0\rangle, |1\rangle \rbrace$ forms the basis along z-direction, from the rotational invariance of singlet state under local unitary operation, it can also be represented as follows:
\begin{align}
|\psi^-\rangle = \frac{1}{\sqrt{2}} (|\psi^j\rangle ~|\psi^j_{\bot}\rangle -|\psi^j_{\bot}\rangle ~|\psi^j\rangle), 
\label{singlet}
\end{align}
where  $|\psi^j\rangle=\cos(\frac{\theta}{2})|0\rangle + \exp(i \phi_j) \sin(\frac{\theta}{2}) |1\rangle$ and $|\psi^j_{\bot}\rangle=-\sin(\frac{\theta}{2}) |0\rangle + \exp(i \phi_j) \cos(\frac{\theta}{2}) |1\rangle$, ($0\leq\theta\leq\pi, 0\leq\phi_j\leq 2\pi$) represent complementary pure states lying on the surface of a Bloch sphere. Now Bob selects either $|\psi^j\rangle$ or $|\psi^j_{\bot}\rangle$ which he intends to remotely prepare at Alice's side. If he performs projective measurement from the basis $\lbrace |\psi^j\rangle, |\psi^j_{\bot}\rangle \rbrace$ on his particle by using the knowledge of the qubit and gets the outcome $|\psi^j_{\bot}\rangle$, then he will communicate the output of his measurement to Alice through some classical channel. Hence, without transferring the particle physically, an unknown qubit will be prepared at Alice's end. Similarly, depending upon the outcome $|\psi^j\rangle$, $|\psi^j_{\bot}\rangle$ will be prepared at Alice's end. This occurs due to the consumption of 1 ebit and 1 cbit of classical communication(CC). Let us call the CC corresponding to $|\psi^j\rangle$ as 'up' and the CC corresponding to $|\psi^j_{\bot}\rangle$ as 'down'. This requirement, utilizing quantum advantage, can not be further reduced~\cite{Pati}. If Bob wants to prepare $|\psi^j\rangle$ at Alice's lab, half of the times he will be successful, though by applying $\sigma_z$ locally, Alice can regenerate $|\psi^j\rangle$ from $|\psi^j_{\bot}\rangle$ for the other half of the times when $|\psi^j\rangle$ is chosen only from the equatorial circle of the Bloch sphere (i.e. $\theta=\frac{\pi}{2}$). Physically, universal operation of complementarity on an unknown qubit is impossible to perform \cite{Buzek}. Therefore, for the states chosen from the circles other than the equator of the Bloch sphere (i.e. $\theta \neq \frac{\pi}{2}$), the process is 50\% successful~\cite{Bennett1}. 

The preparation of a remote state $|\psi^j\rangle$ at Alice's lab by using singlet state $|\psi^-\rangle$ shared between Alice and Bob depends on the representation of $|\psi^-\rangle$ in terms of two complementary pure states $|\psi^j\rangle$ and $|\psi^j_{\bot}\rangle$ as given by Eq.(\ref{singlet}). However, similar representations for other three Bell states exist~\cite{Pati}.  These are as follows,
\begin{align}
|\psi^+\rangle &= \frac{1}{\sqrt{2}} (|01\rangle + |10\rangle)
= -\frac{1}{\sqrt{2}} [\sigma_z |\psi^j_{\bot}\rangle ~|\psi^j\rangle - \sigma_z |\psi^j\rangle ~|\psi^j_{\bot}\rangle],\\
|\Phi^+\rangle &= \frac{1}{\sqrt{2}} (|00\rangle + |11\rangle)
= \frac{1}{\sqrt{2}} [i \sigma_y |\psi^j_{\bot}\rangle ~|\psi^j\rangle - i \sigma_y |\psi^j\rangle ~|\psi^j_{\bot}\rangle]
\end{align}
and
\begin{align}
|\Phi^-\rangle &= \frac{1}{\sqrt{2}} (|00\rangle - |11\rangle)
= \frac{1}{\sqrt{2}} [\sigma_x |\psi^j_{\bot}\rangle ~|\psi^j\rangle - \sigma_x |\psi^j\rangle ~|\psi^j_{\bot}\rangle]
\end{align}
where $\sigma_x, \sigma_y$ and $\sigma_z$ are Pauli matrices in x-, y- and z- basis respectively. Bob performs projective measurement in the basis $\lbrace |\psi^j\rangle, |\psi^j_{\bot}\rangle\rbrace$ and communicates the result to Alice. If Bob obtains $|\psi^j_{\bot}\rangle$, then Alice applies local unitaries $\sigma_z, i\sigma_y$ and $\sigma_x$ for shared states $|\psi^+\rangle, |\Phi^+\rangle$ and $|\Phi^-\rangle$ respectively, instead of $\mathbb{I}_2$ for singlet state $|\psi^-\rangle$ so that $|\psi^j\rangle$ is remotely prepared at Alice's lab. On the other hand, if Bob obtains $|\psi^j\rangle$, then Alice applies local unitaries $\mathbb{I}_2, i\sigma_z \sigma_y$ and $\sigma_z \sigma_x$ for shared states $|\psi^+\rangle, |\Phi^+\rangle$ and $|\Phi^-\rangle$ respectively, instead of $\sigma_z$ for singlet state $|\psi^-\rangle$ so that a state $|\psi^j\rangle$ from the equatorial circle of the Bloch sphere is remotely prepared at Alice's lab. Therefore Alice's operation is invariant under local unitary for deterministic preparation of an unknown remote state at Alice's side when distinct Bell states are shared initially.

If the state is prepared with certainty, then the measure of fidelity between the target state and prepared state will be 1. Consider for a particular bipartite state, the prepared state is $\rho^p$ and the pure state, $|\psi^d\rangle$ is desired to be prepared, then by definition, the generalised symmetric function of fidelity between these two quantum states can be expressed as~\cite{Popescu},
\begin{align*}
f(\rho^p,\rho^d)= f(\rho^d,\rho^p)= \Big(\operatorname{Tr}\Big[\sqrt{\sqrt{\rho^d}.\rho^p.\sqrt{\rho^d}}\Big]\Big)^2 = \langle \psi^d|\rho^p|\psi^d \rangle
\end{align*}
where $\rho^d=|\psi^d\rangle\langle\psi^d|$. The local unitary operation on Alice, i.e., from $\lbrace \mathbb{I}_2, \sigma_z \rbrace$, suffices to be considered for the equator of the Bloch sphere, because the rotational freedom of choosing the desired state by the sender may be accompanied by similar rotation on the prepared state by the receiver, i.e., $f(\rho^p,\rho^d) = f(U\rho^p U^{\dagger},U\rho^d U^{\dagger})$ for all unitary operators $U$. For a singlet state with projective measurements performed on Bob's particle, a remote state can be prepared at Alice's lab perfectly such that $f=1$. This is the condition for deterministic RSP~\cite{DRSP}. Since Alice has no information available to her about the state to be prepared, the average fidelity can be calculated by considering all input states from the circle on the Bloch sphere (taking single infinity of bits into account) as,
\begin{equation}
f_{av} = \frac{1}{2\pi} \int_0^{2\pi} f ~d\phi_i, ~~~~~(i=1,2,3,...).
\end{equation}

\subsection{Optimal fidelity for classical strategy}

Let us suppose that Bob wants to prepare an unknown pure qubit state $|\psi^d\rangle=\cos(\frac{\theta}{2})|0\rangle+e^{i\phi^d} \sin(\frac{\theta}{2})|1\rangle$ ($0\leq\theta\leq\pi, 0\leq\phi^d\leq 2\pi$) in Alice's lab without transferring the particle physically. $|\psi^d\rangle$ is known to Bob but unknown to Alice. $|\psi^d\rangle$ lies on the boundary of a circle (with $\theta$ fixed) which is perpendicular to the z-axis, specified a priori to both Alice and Bob. Here, Bob is not allowed to exploit any quantum resource that can be initially shared between Alice and Bob. We consider that Alice and Bob are completely uncorrelated and do not even share a separable state between them. Bob is only allowed to make use of classical channel by which he can send 1 cbit of information to Alice, and Alice is free to prepare a qubit upto local operations. If the prepared state in Alice's lab is $\rho^p$, then the closeness of the desired state and the prepared state is $\langle \psi^d|\rho^p|\psi^d \rangle$, which is $|\langle\psi^d|\psi^p\rangle|^2$ for pure state $\rho^p=|\psi^p\rangle\langle\psi^p|$. The classical fidelity can be calculated by averaging the closeness over infinitely many runs where in each run Bob is given different $|\psi^d\rangle$ from a certain plane of the Bloch sphere with fixed $\theta$ consistent with the protocol of RSP.

If the classical communication (CC) is not allowed from Bob to Alice, then Alice has to randomly guess the desired state which will either match or does not match with the desired state. Hence, the fidelity becomes $\frac{1}{2}$ which is the lower bound of classical fidelity in any circumstances. If CC is allowed from Bob to Alice, then we consider a specific classical strategy for preparing the remote states at Alice's side. Under a given scenario, the optimality of a classical strategy is in general very difficult to prove. However, it is more important to clearly and formally describe the scenario where we have defined the task. The scenario offers us the ground on which we can compare the two strategies with and without a quantum mechanical resource. In comparison with the quantum strategy discussed in the previous subsection, we restrict the scenario where Bob does qubit measurement on his particle and sends the outcome of the measurement via classical channel to Alice and Alice applies the unitary from $\lbrace \sigma_z,\mathbb{I}_2 \rbrace$ on her particle to recover the state. Note that any shared randomness between Alice and bob is not allowed in our classical model, which is demonstrated as follows.

Here Bob measures a dichotomic observable $\vec{n_B}.\vec{\sigma}$ where $\vec{n_B} = (\sin \theta_B \cos \phi_B, \sin \theta_B \sin \phi_B, \cos \theta_B)$ with $0\leq\theta_B\leq\pi, 0\leq\phi_B\leq 2\pi$ and sends the outcome (either up or down) to Alice by a classical channel. This classical strategy is consistent with the probabilistic preparation or prepare and measure scenario where Bob applies the best possible measurement to make the fidelity higher. It is not possible for Bob to share his knowledge about the state via CC to Alice by doing such measurement in arbitrary direction as Alice and Bob are perfectly uncorrelated. The probability of getting up(down) outcome is given by $\langle \psi^d|\frac{\mathbb{I}_2 \pm \vec{n_B}.\vec{\sigma}}{2}|\psi^d\rangle = \operatorname{Tr}[(\frac{\mathbb{I}_2 \pm \vec{n_B}.\vec{\sigma}}{2}).|\psi^d\rangle\langle\psi^d|]$. Now Alice can prepare either $|\psi^p_1\rangle$ or $|\psi^p_2\rangle$ depending upon the outcomes either up or down. Note that, Alice prepares the states by using the information of $\theta$, known to both Alice and Bob beforehand. However, the information of $\phi^d$ ($0\leq\phi^d\leq 2\pi$) remains completely unknown to Alice even after getting the CC from Bob. These states can be prepared as follows.

Alice chooses $|\psi^p_2\rangle = \cos(\frac{\theta}{2})|0\rangle+e^{i\phi^p_2} \sin(\frac{\theta}{2})|1\rangle$ ($0\leq\phi^p_2\leq 2\pi$) and applies either $\sigma_z$ and $\mathbb{I}_2$ locally depending upon the outcome either up and down respectively or vice-versa. Here we take $|\psi^p_1\rangle\langle\psi^p_1| = \sigma_z.|\psi^p_2\rangle\langle\psi^p_2|.\sigma_z$ because unitary evolution preserves the purity of $|\psi^p_1\rangle$. Alice fixes the polar angle ($\theta$) same as that of Bob using the information known beforehand and single infinity bits of information (i.e. $\phi^d$) remains unknown to her. As $|\psi^p_2\rangle$ is arbitrarily chosen from the given circle of the Bloch sphere, hence the direction of $|\psi^p_1\rangle$ is also arbitrary.

Now, as the input state is unknown to Alice, hence by taking average over all the input states $|\psi^d\rangle$ from the specified circle of the Bloch sphere (with $\theta$ fixed), the fidelity expression for a classical strategy becomes
\begin{align}
f_{cl} =& \frac{1}{2\pi} \int_{0}^{2\pi} \Big(\langle \psi^d|\frac{\mathbb{I}_2+\vec{n_B}.\vec{\sigma}}{2}|\psi^d\rangle ~\langle\psi^d|\psi^p_1\rangle\langle\psi^p_1|\psi^d\rangle \nonumber\\
&+ \langle \psi^d|\frac{\mathbb{I}_2-\vec{n_B}.\vec{\sigma}}{2}|\psi^d\rangle ~\langle\psi^d|\psi^p_2\rangle\langle\psi^p_2|\psi^d\rangle\Big) ~d\phi^d.
\end{align}
The classical fidelity can be optimized with respect to the measurement parameters chosen by Bob and the parameters of the states prepared by Alice. In this strategy, the sharing of the knowledge of Bloch co-ordinate system in the $|\psi^d\rangle$-plane (i.e. the plane with fixed polar angle $\theta$ where $|\psi^d\rangle$ lies) between Alice and Bob is redundant except the need to represent the desired and the prepared state in the same basis and in the same Bloch reference frame to compute the fidelity.

Corresponding to the given condition, it can be easily seen that,
\begin{align}
f_{cl} &= \frac{3}{4} + \frac{1}{4} [\cos 2\theta-\cos(\phi_B - \phi^p_2) \sin^3\theta ~\sin\theta_B] \nonumber\\
&\leq \frac{3}{4} + \frac{\cos 2\theta + \sin^3 \theta}{4}.
\label{fcl}
\end{align}
Here the inequality comes from the maximization over $\theta_B,\phi_B,\phi^p_2$ and the equality holds for $\theta_B= \frac{\pi}{2}, \phi_B - \phi^p_2 = \pm \pi$. Note that, when $|\psi^p_1\rangle$ and $|\psi^p_2\rangle$ are not related with each other by $\sigma_z$, then also the upper bound given by the inequality(\ref{fcl}) remains the same\footnote{If Alice chooses two pure states, e.g. $|\psi^p_1\rangle = \cos(\frac{\theta}{2})|0\rangle+e^{i\phi^p_1} \sin(\frac{\theta}{2})|1\rangle$ ($0\leq\phi^p_1\leq 2\pi$) and $|\psi^p_2\rangle = \cos(\frac{\theta}{2})|0\rangle+e^{i\phi^p_2} \sin(\frac{\theta}{2})|1\rangle$ ($0\leq\phi^p_2\leq 2\pi$) randomly from the specified circle with fixed polar angle, $\theta$ of the Bloch sphere depending upon the CC from Bob, then we have $f_{cl} = \frac{3}{4} + \frac{\cos 2\theta}{4} + \frac{\sin^3\theta ~\sin\theta_B}{8} [\cos(\phi_B - \phi^p_1)-\cos(\phi_B - \phi^p_2) ] \leq \frac{3}{4} + \frac{\cos 2\theta + \sin^3 \theta}{4}$ where the maximum occurs when $\theta_B= \frac{\pi}{2}, \phi_B - \phi^p_1 = 0 ~\text{or}~ 2\pi, \phi_B - \phi^p_2 = \pm \pi$.}.

Hence, the classical upper limit of fidelity for preparing an unknown state becomes $f_{cl}^{\max} = \frac{3}{4} + \frac{\cos 2\theta + \sin^3 \theta}{4}$ which goes beyond $\frac{1}{2}$ by utilizing 1 cbit of information. It can be easily checked that, this limit can not be improved by using a Positive Operator Valued Measurement(POVM) instead of a projective measurement at Bob's side. 

Unlike the standard scheme of teleportation~\cite{Massar}, the protocol of RSP achieves higher optimal limit of classical fidelity, i.e., $\frac{3}{4} \leq f_{cl}^{\max} \leq 1$ by considering all possible circles of latitude $\theta$ ($0\leq \theta \leq \pi$) on the Bloch sphere. For example, when the equatorial circle with $\theta= \frac{\pi}{2}$ is considered, then $f_{cl}^{\max} = \frac{3}{4}$, when the non-equatorial circle with $\theta= \frac{\pi}{4} ~\text{or}~ \frac{3\pi}{4}$ is considered, then $f_{cl}^{\max} = 0.838$, and when the two poles with $\theta= 0 ~\text{or}~ \pi$ is considered, then $f_{cl}^{\max} = 1$. The choice of states from the two poles of the Bloch sphere leaves no uncertainty in predicting the state classically because these correspond to two particular states whose polar angles are already pre-shared between Alice and Bob. 

$f_{cl}^{\max}$ under the framework of RSP for the choice of different circles of the Bloch sphere are plotted in Fig \ref{fidelity} where it is seen that $f_{cl}^{\max}(\theta)$ is an even function of $\theta$ w.r.t. $\theta=\frac{\pi}{2}$. $f_{cl}^{\max}$ is symmetrically distributed on either side of the equatorial circle of the Bloch sphere and gets amplified as the size of the non-equatorial circle contracts over the Bloch sphere. Any shared state without making use of quantum resource can achieve such fidelity. 

\begin{figure}[!ht]
\centering
\includegraphics[width=0.75\linewidth]{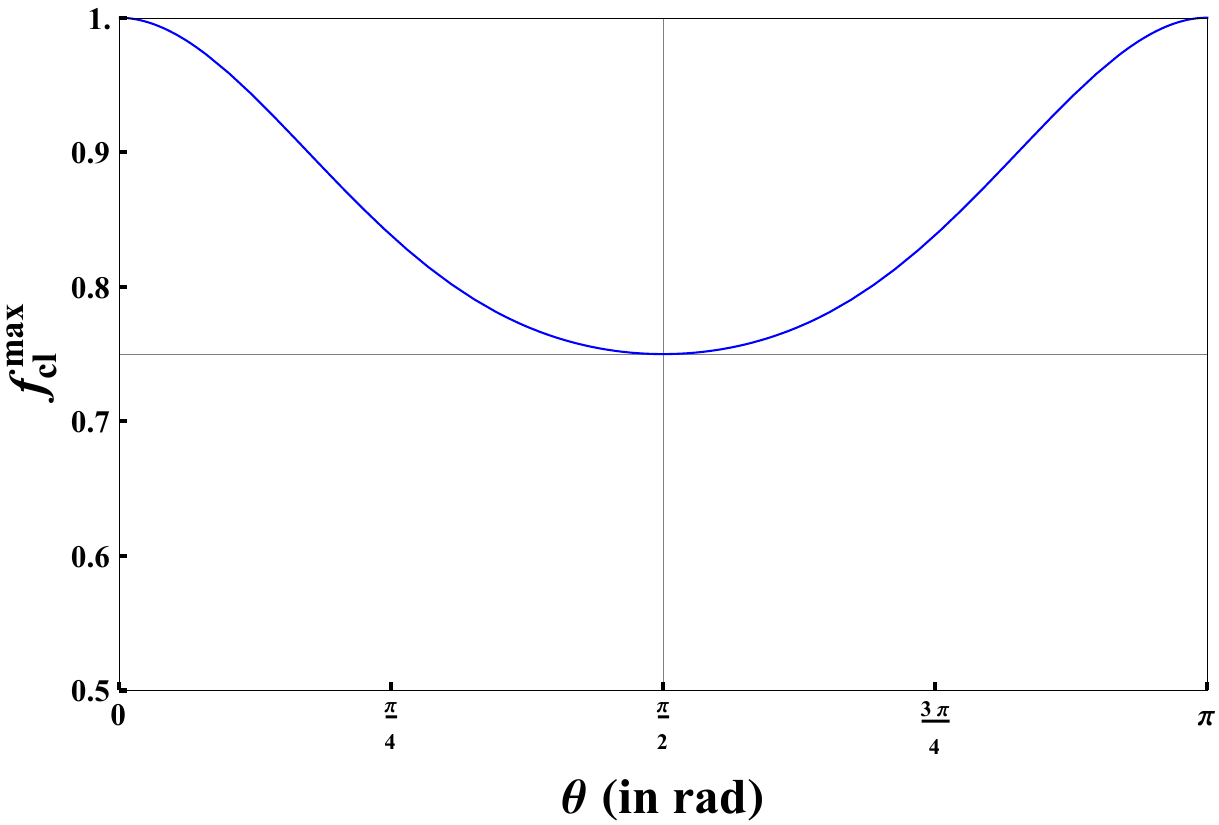}
\caption{\footnotesize $f_{cl}^{\max}$ is plotted against the polar angle $\theta$ of a fixed circle perpendicular to the z-axis of the Bloch sphere from where Bob chooses a state to be remotely prepared at Alice's side}
\label{fidelity}
\end{figure}

Our classical strategy is useful to distinguish the classical and quantum regimes under the given RSP-protocol where the measurement outcomes are sent via classical channel keeping the azimuthal angles of the qubits completely oblivious to the receiver. The violation of the upper bound of the classical fidelity implies a non-classical advantage in a particular way. However, one can construct other classical strategies in different contexts. For example, Horodecki et al proposed a classical model~\cite{Horodecki14} in a deterministic preparation scenario where the classical channel is utilized in a different manner by extracting the information of $\phi^d$ partially with the help of a Bloch reference frame. Despite the optimality of such strategy in terms of RSP from the equatorial great circle of the Bloch sphere, it is not obvious to exactly generalise the optimality for the non-equatorial circles of the Bloch sphere. In our work, we consider both the classical and quantum strategies consistent with the prepare and measure scenario where 1 cbit information arises from the outcomes of a qubit measurement which affects the system non-trivially. Our strategy works well for any circle of the Bloch sphere, induced by the sequential qubit measurements scenario, as discussed below.


\section{Definition and functionality of the scenario of RSP by multiple sequential observers} \label{R2.3}

Here we consider a scenario where half of a single-copy of a bipartite entangled pair is possessed by Alice at the receiver's end and another half is sequentially possessed by a number of independent Bobs (i.e. $\lbrace$Bob$^i\rbrace_i$) at the senders' end and the task of each Bob is to prepare remote states at Alice's side reliably with non-classical advantage by making use of unsharp measurements and by choosing the remote states from a specific circle perpendicular to the z-axis of the Bloch sphere. The polar angle of the remote states chosen by all the Bobs remains same throughout the protocol and is known to all the Bobs and Alice, while the azimuthal angle may be arbitrarily chosen by each Bob keeping it completely hidden to Alice. 

During the course of the protocol, Alice does not make measurements to decipher the state prepared by any of the Bobs. Alice is left to decide later which particular state (prepared by the corresponding Bobs) she wants to utilize. Such a scenario can be relevant in certain secret sharing protocols involving multiple parties, where Alice may not want to decide or reveal in advance her choice of the particular Bob whose supplied information she would like to utilize. Note that our protocol does not allow shared randomness among the sequence of Bobs, thereby ruling out any  possibility of sharing correlated multipartite mixed states among them. The optimal classical fidelity for such a protocol is thus determined accordingly. The upper bound on the fidelity obtained through our classical strategy is therefore, the same for all Bobs in the multiple observer scenario, and is given by the bound obtained in the single observer scenario as discussed in the previous subsection. 

Initially the state, $\rho^1$ is shared between Alice and Bob$^1$. If Bob$^1$ does a projective or sharp measurement on the basis known to him, then by sending the outcome of measurement to Alice, Bob can prepare the desired state at Alice's side with certainty. This leads to the complete loss of entanglement between Alice and Bob$^1$. However, unsharp measurement does not completely destroy the entanglement between them, and hence, the possibility of further utilization of the state remains by exploiting the remaining quantum resource, since the remaining entanglement guarantees the state to be discordant~\cite{Costa}. However, Bob$^1$ can not measure arbitrarily weakly as in this case, the fidelity of the prepared state at Alice's side would be very low due to a trade-off between sharpness and fidelity.

When the desired remote state is chosen from the equatorial circle of the Bloch sphere (i.e. $\theta=\frac{\pi}{2}$), the local unitary operation from $\lbrace \sigma_z,\mathbb{I}_2 \rbrace$ at Alice's side depending upon the CC from Bob$^i$ ($i=1,2,3,...$) can restore the desired state at Alice's side with finite average RSP-fidelity. There is no need to post-select Alice's subsystem to achieve this corresponding to a particular outcome of the unsharp measurement done by Bob$^i$. Therefore, the preparation of remote states from the equatorial great circle of the Bloch sphere corresponds to 100\% successful RSP-protocol.

If the desired state $|\psi^i\rangle$ is not chosen from the equatorial circle of the Bloch sphere (i.e. $\theta \neq \frac{\pi}{2}$, $0< \theta < \pi$),  the conversion $|\psi^i_{\bot}\rangle \rightarrow |\psi^i\rangle ~\forall i$ or vice-versa, is impossible for Alice to implement under any unitary operation without knowing $|\psi^i\rangle$. However, the measurement statistics of the desired state can still be reproduced at Alice's side by assuming the swap between the outcomes, i.e. (up $\rightarrow$ down) and (down $\rightarrow$ up), {\it a priori}. Here Alice allows the state without rotating it when she receives the down outcome via CC from Bob$^i$ and does not consider the CC for preparing the desired state (i.e. $|\psi^i\rangle$) corresponding to the up outcome communicated by Bob$^i$ as she is aware from the knowledge of the polar angle $\theta$, that she can not deterministically prepare $|\psi^i\rangle$ by applying $\sigma_z$. When Alice rejects the state  without altering her subsystem for the up outcome,  she takes our classical strategy by picking a random pure state from the fixed non-equatorial circle of the Bloch sphere which gives a maximum of $f_{cl}^{\max}(\theta \neq \frac{\pi}{2})$ as the fidelity of preparing an unknown state. On the other hand, Alice uses the quantum strategy to evaluate the RSP-fidelity corresponding to the down outcome. During preparation of $|\psi^i_{\bot}\rangle$ from a specified circle of the Bloch sphere, the action of Alice gets reversed depending upon the outcome of measurement up/down as communicated by a Bob. Moreover, Alice does not know at any instant whether she and Bob$^i$ share a quantum channel or not, and hence, she applies the above post-selection method uniformly for all the Bobs corresponding to RSP from non-equatorial plane of the Bloch sphere. The occurrence of selecting or rejecting the state by Alice is associated with probability $\frac{1}{2}$. Hence, this scenario leads to the RSP of 50\% success.


We categorize RSP in the following two ways:
\begin{table}[h!]
\centering
 \begin{tabular}{| c | c |} 
 \hline
 Category-A & Category-B \\ 
 \hline
  Remote states chosen from the equatorial & Remote states chosen from a non-equatorial \\
  plane of the Bloch sphere bounded by the & plane of the Bloch sphere bounded by a \\ 
  great circle positioned perpendicular to the & circle positioned perpendicular to the \\
  z-axis where all the states have the same & z-axis where all the states have the same \\
  polar angle $\theta=\frac{\pi}{2}$ & polar angle $\theta \neq \frac{\pi}{2}$ \\
 \hline
 RSP is 100\% successful here because the & RSP is 50\% successful here because the \\
 conversion between the complementary states & conversion between the complementary states \\
 is possible via unitary (i.e. $\sigma_z$) & is impossible via unitary (e.g. $\sigma_z$) \\
 \hline
 Post-selection need not  be implemented here & Post-selection  needs to be implemented here \\
 \hline
\end{tabular}
\end{table}

Next, Alice reverses her previous unitary operation to reinstate the shared state to be used by the next Bob for the same task. She applies $\sigma_z$ if her previous operation was $\sigma_z$ (in case of remote states from the equatorial circle of the Bloch sphere), and does nothing if it was $\mathbb{I}_2$. The reverse operation is not performed for the remote states chosen from a non-equatorial circle of the Bloch sphere. 

In a nutshell, Bob$^1$ starts the process and passes the particle in his possession to Bob$^2$ such that Bob$^2$ can prepare the remote state at Alice's lab after the suitable reverse operation done by Alice. Since Bob$^2$ is independent of Bob$^1$, he considers the average effect of all probable actions taken by Bob$^1$. During this second step, Bob$^2$ performs an unsharp measurement on the resultant state, and makes a CC to Alice.  Alice again applies a suitable unitary correction to determine the average RSP-fidelity depending upon the outcome of measurement done by Bob$^2$ and the chosen plane of the Bloch sphere with known $\theta$. Again, the suitable reverse operation, (i.e., $\lbrace \sigma_z,\mathbb{I}_2 \rbrace$) for $\theta=\frac{\pi}{2}$ is done by Alice to leave the state to be re-utilised by Bob$^3$. The method of post-selection is  applied when $\theta \neq \frac{\pi}{2}$.  

The above process is repeated for all the following Bobs. In this way, all independent $n$-number of Bobs  can sequentially access the particle  and can prepare remote states at Alice's end by making use of a single copy of the initial entangled state. Fig.\ref{scheme} depicts the entire protocol in the aforementioned scenario. The process continues as long as Bobs are able to perform the task of RSP with average non-classical fidelity which attains maximum value if only the last Bob in the sequence performs sharp measurement. Our aim is to find the bound on the maximum number of Bobs in the framework of the above RSP  scenario.

\begin{figure}[!ht]
\centering
\includegraphics[width=0.75\linewidth]{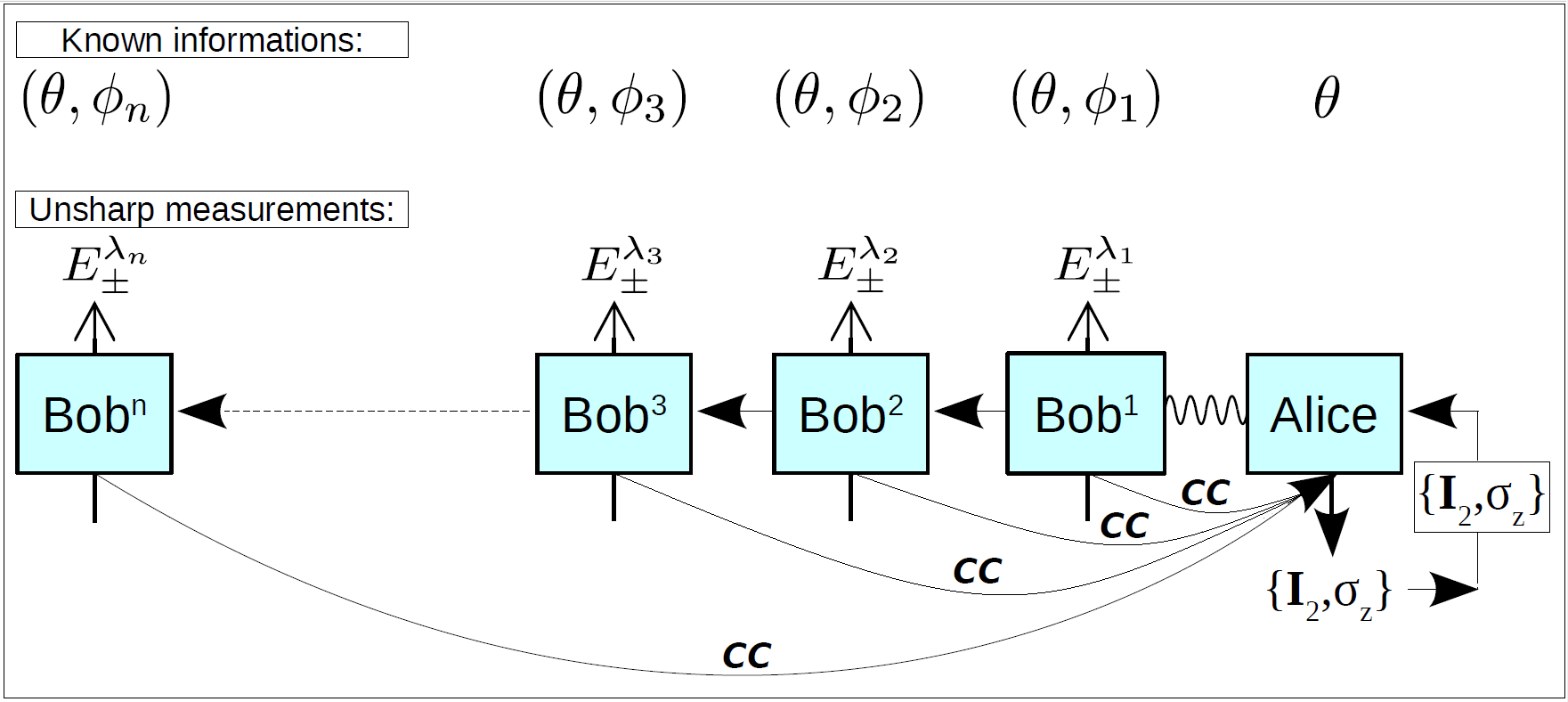}
\caption{\footnotesize Schematic diagram of remote state preparation at Alice's lab by multiple Bobs which are independent of each other. The operations, $\lbrace \mathbb{I}_2, \sigma_z \rbrace$ are applicable for remote states from the equatorial circle of the Bloch sphere.}
\label{scheme}
\end{figure}



In general, an unsharp measurement which is a class of a POVM is characterised by the effect operators~\cite{Busch}, $E_{\pm}^{\lambda_i} = \lambda_i P^i_{\pm} + (1-\lambda_i) \frac{\mathbb{I}_2}{2}$, ~($0 \leq \lambda_i \leq 1; \sum_{a=+,-} E_a^{\lambda_i} = \mathbb{I}_2$) corresponding to outcomes $\lbrace +,- \rbrace$ and $\lambda_i$ corresponds to the sharpness of measurement performed by Bob$^i$. The choice of measurements corresponding to projectors, $P^i_{+} = |\psi^i\rangle\langle\psi^i|$ and $P^i_{-} = |\psi_{\bot}^i\rangle\langle\psi_{\bot}^i|$ are distinct for each Bob. Depending on the outcome $a$, i.e., $\lbrace$up(+) or down(-)$\rbrace$, the normalised conditional state at Alice's side becomes,
\begin{equation}
\rho^i_{A|E_a^{\lambda_i}} = \frac{\operatorname{Tr}_{\text{Bob}^i}\Big[\Big(\mathbb{I}_2 \otimes \sqrt{E_a^{\lambda_i}}\Big) ~\rho^i ~\Big(\mathbb{I}_2 \otimes \sqrt{E_a^{\lambda_i}}\Big)\Big]}{\operatorname{Tr}\Big[\Big(\mathbb{I}_2 \otimes \sqrt{E_a^{\lambda_i}}\Big) ~\rho^i ~\Big(\mathbb{I}_2 \otimes \sqrt{E_a^{\lambda_i}}\Big)\Big]}, ~(i=1,2,3,...).
\label{cond}
\end{equation}
The probability of getting such state is given by, $p^i_a=\operatorname{Tr}[(\mathbb{I}_2 \otimes \sqrt{E_a^{\lambda_i}}) ~\rho^i ~(\mathbb{I}_2 \otimes \sqrt{E_a^{\lambda_i}})]$, where $a\in\lbrace +,- \rbrace$. 


Bob$^i$ communicates the outcome of his measurement to Alice through some classical channel and Alice applies suitable unitary operation on her conditional state. Hence the average RSP-fidelity of Alice's state prepared by communicating either of the outcomes $\lbrace +,- \rbrace$ from Bob$^i$, and by considering all the input states from the circle of Bloch sphere chosen by Bob$^i$, is given by
\begin{align}
f_{av^{+}}^{AB^i} &= \frac{1}{2\pi} \int_0^{2\pi} \langle \psi^i_{\bot} | \rho^i_{A|E_{+}^{\lambda_i}} | \psi^i_{\bot} \rangle ~d\phi_i, \nonumber\\
f_{av^{-}}^{AB^i} &= \frac{1}{2\pi} \int_0^{2\pi} \langle \psi^i | \rho^i_{A|E_{-}^{\lambda_i}} | \psi^i \rangle ~d\phi_i, ~~~~~(i=1,2,3,...)
\label{afid}
\end{align}
as $\phi_i$ ($0\leq \phi_i \leq 2\pi$) remains completely unknown to Alice. The first and second equality corresponds to the cases where each Bob desires to prepare $| \psi^i_{\bot} \rangle$ or $| \psi^i \rangle$ from the chosen circle of the Bloch sphere, respectively. Whereas depending upon Alice's course of action as described before, the average RSP-fidelity by considering both the outcomes of unsharp measurement done by Bob$^i$ takes the form,
\begin{align}
f_{av}^{AB^i} =& \frac{1}{2\pi} \int_0^{2\pi} (p_{+}^i ~\langle \psi^i | \sigma_z. \rho^i_{A|E_+^{\lambda_i}}. \sigma_z | \psi^i \rangle \nonumber\\
&+ p_{-}^i ~\langle \psi^i | \mathbb{I}_2. \rho^i_{A|E_-^{\lambda_i}}. \mathbb{I}_2 | \psi^i \rangle) ~d\phi_i
\label{psi100}
\end{align}
in case of the preparation of $|\psi^i\rangle$ from the equatorial circle of the Bloch sphere with 100\% successful RSP or,
\begin{align}
f_{av}^{AB^i} =& \frac{1}{2\pi} \int_0^{2\pi} (p_{+}^i ~f_{cl}^{\max} + p_{-}^i ~\langle \psi^i | \rho^i_{A|E_{-}^{\lambda_i}} | \psi^i \rangle) ~d\phi_i
\label{psi50}
\end{align}
in case of the preparation of $|\psi^i\rangle$ from a non-equatorial circle of the Bloch sphere with 50\% successful RSP or,
\begin{align}
f_{av}^{AB^i} =& \frac{1}{2\pi} \int_0^{2\pi} (p_{+}^i ~\langle \psi^i_{\bot} | \mathbb{I}_2. \rho^i_{A|E_+^{\lambda_i}}. \mathbb{I}_2 | \psi^i_{\bot} \rangle \nonumber\\
&+ p_{-}^i ~\langle \psi^i_{\bot} | \sigma_z. \rho^i_{A|E_-^{\lambda_i}}. \sigma_z | \psi^i_{\bot} \rangle) ~d\phi_i
\end{align}
\label{psi_perp100}
in case of the preparation of $|\psi^i_{\bot}\rangle$ from the equatorial circle of the Bloch sphere with 100\% successful RSP or,
\begin{align}
f_{av}^{AB^i} = \frac{1}{2\pi} \int_0^{2\pi} (p_{+}^i ~\langle \psi^i_{\bot} | \rho^i_{A|E_{+}^{\lambda_i}} | \psi^i_{\bot} \rangle + p_{-}^i ~f_{cl}^{\max}) ~d\phi_i
\label{psi_perp50}
\end{align}
in case of the preparation of $|\psi^i_{\bot}\rangle$ from a non-equatorial circle of the Bloch sphere with 50\% successful RSP. Here, $p_a^i$ ($a\in\lbrace +,- \rbrace$) denotes the probability of getting the outcome $a$.


After Bob$^i$ sends his particle to Bob$^{i+1}$ to perform the task of RSP,  Bob$^{i+1}$ being completely uninformed about the choice made by Bob$^i$ before his measurement, considers the state shared with Alice as an average over all possible input states chosen by Bob$^i$ and all possible measurement outcomes. As Alice's operation is reversible, hence according to L\"{u}der transformation (non-selective), the shared state between Alice and Bob$^{i+1}$ becomes,
\begin{align}
\rho^{i+1} = \frac{1}{2\pi} \int_0^{2\pi} \sum_{a=+,-} \Big(\mathbb{I}_2 \otimes \sqrt{E_{a}^{\lambda_i}}\Big) ~\rho^i ~\Big(\mathbb{I}_2 \otimes \sqrt{E_{a}^{\lambda_i}}\Big) ~d\phi_i,& \nonumber\\
(i=1,2,3,...).&
\label{gs}
\end{align}
This is also true for $\theta \neq \frac{\pi}{2}$ due to the independence of the Bobs and the ignorance of Bob$^{i+1}$ about the direction and outcome of the measurements done by the previous Bobs. Using the shared state $\rho^{i+1}$, Bob$^{i+1}$ completes his task, resulting in the conditional state given by Eq.(\ref{cond}) with $i \rightarrow i+1$.


The above steps are recursive. For the n-th Bob, $f_{av}^{AB^n}$ (n=1,2,3,...) happens to be a function of all $\lambda_i$ (i=1,..,n). The entire process stops when the average fidelity goes below the classical limit of fidelity for all possible range of $\lambda_i$s compatible for successful remote state preparation.

\section{Upper bound on the number of Bobs} \label{R3}


Now, we are in a position to explore the bound on the number of Bobs who can sequentially prepare remote states, picked from the different circles of a Bloch sphere, at Alice's side with non-classical advantage compared to our classical strategy. Note that such a bound depends on the framework where the task is defined. We assume the initial state between Alice and Bob$^1$ to be one of the maximally entangled states, i.e., the singlet state, since the resource in terms of quantum geometric discord is maximum (i.e., 1) for such a state~\cite{Discord,Luo}.  Bob$^1$ chooses to perform $\lbrace E_{\pm}^{\lambda_1} \rbrace$, and by communicating the output classically, he wants to prepare states from $\lbrace |\psi^1\rangle, |\psi^1_{\bot}\rangle \rbrace$ to Alice.

\begin{lemma}
The fidelity of preparing any pure state to Alice subject to unsharp measurement by Bob$^1$ with sharpness $\lambda_1$ while sharing a Bell state is $\frac{1+\lambda_1}{2}$.
\end{lemma}

\begin{proof}
We consider the initial state $\rho^1=|\psi^-\rangle\langle\psi^-|$, which is shared between Alice and Bob$^1$. Depending upon the outcome $\lbrace +,- \rbrace$ obtained by Bob$^1$, the state produced at Alice becomes either
\begin{equation}
\rho^1_{A|E_+^{\lambda_1}} = \lambda_1 |\psi^1_{\bot}\rangle\langle\psi^1_{\bot}| + \frac{1-\lambda_1}{2} \mathbb{I}_2, ~~~\forall \theta,\phi_1
\end{equation}
or
\begin{equation}
\rho^1_{A|E_-^{\lambda_1}} = \lambda_1 |\psi^1\rangle\langle\psi^1| + \frac{1-\lambda_1}{2} \mathbb{I}_2, ~~~\forall \theta,\phi_1.
\end{equation}

Therefore, it follows that, $f^{AB^1} = \langle \psi^1_{\bot} | \rho^1_{A|E_{+}^{\lambda_1}} | \psi^1_{\bot} \rangle = \langle \psi^1 | \rho^1_{A|E_{-}^{\lambda_1}} | \psi^1 \rangle = \frac{1+\lambda_1}{2}, ~\forall \theta,\phi_1$.
\end{proof}

Corresponding to the completely successful preparation of remote states $\lbrace |\psi^1\rangle \rbrace_{\phi_1}$ from the equatorial circle of the Bloch sphere (i.e., with $\theta=\frac{\pi}{2}$), the average RSP-fidelity, $f_{av}^{AB^1} = \frac{1}{2\pi} \int_0^{2\pi} (p_{+}^1 ~\langle \psi^1 | \sigma_z.\rho^1_{A|E_{+}^{\lambda_1}}.\sigma_z | \psi^1 \rangle + p_{-}^1 ~\langle \psi^1 | \rho^1_{A|E_{-}^{\lambda_1}} | \psi^1 \rangle) ~d\phi_1 = \frac{1+\lambda_1}{2} > \frac{3}{4}$ occurs when $\lambda_1 > \frac{1}{2}$ in comparison with our classical strategy. Here, the probability of getting $\rho^1_{A|E_{\pm}^{\lambda_1}}$, i.e. $p_{\pm}^1$ satisfies $p_{+}^1=p_{-}^1=\frac{1}{2} ~\forall \phi_1,\lambda_1$. The result is unaltered for the choice of remote states $\lbrace |\psi^1_{\bot}\rangle \rbrace_{\theta=\frac{\pi}{2},\phi_1}$. 

On the other hand, for 50\% successful preparation of remote states $\lbrace |\psi^1\rangle \rbrace_{\phi_1}$ from any non-equatorial circle of the Bloch sphere (i.e. with $\theta \neq \frac{\pi}{2}$), the average RSP-fidelity, by considering both the outcomes, becomes $f_{av}^{AB^1} = \frac{1}{2\pi} \int_0^{2\pi} (p_{+}^1 ~f_{cl}^{\max} + p_{-}^1 ~\langle \psi^1 | \rho^1_{A|E_{-}^{\lambda_1}} | \psi^1 \rangle) ~d\phi_1 = \frac{f_{cl}^{\max}}{2} + \frac{1+\lambda_1}{4}$ where $p_{+}^1=p_{-}^1=\frac{1}{2} ~\forall \theta,\phi_1,\lambda_1$. This is because Alice post-selects the state corresponding to the down outcome at Bob$^1$'s side and discards the state without considering it for the task of RSP corresponding to the up outcome at Bob$^1$'s side and the fidelity can achieve the classical upper bound $f_{cl}^{\max}=\frac{3}{4} + \frac{\cos 2\theta + \sin^3 \theta}{4}$ when Alice discards the state. Therefore $f_{av}^{AB^1} = \frac{f_{cl}^{\max}}{2} + \frac{1+\lambda_1}{4} > f_{cl}^{\max}$ occurs when $\lambda_1 > \frac{1+\cos 2\theta + \sin^3 \theta}{2}$ compared to our classical strategy (see Fig.\ref{rsps}). By applying similar arguments for the remote states $\lbrace |\psi^1_{\bot}\rangle \rbrace_{\theta \neq \frac{\pi}{2},\phi_1}$, the result remains unchanged. Hence for this region of $\lambda_1$, the joint state remains useful for further utilization in the task of RSP.

After suitable reversible operation done by Alice, Bob$^1$ can now send his particle to Bob$^2$ for the next round of the protocol. We consider that Bob$^i$ prepares either $|\psi^i \rangle$ or $|\psi^i_{\bot}\rangle$ in the $i$-th round of the procedure. Let us suppose that the remote state to be prepared at Alice's side in any round of the procedure, comes from a circle with fixed $\theta$ on the Bloch sphere i.e., a single infinity bits of information ($\theta$) is known to all the senders and the receiver. This implies that once a particular $\theta$ is fixed at the beginning of the procedure, it remains the same throughout the procedure. While the azimuthal angle of the state remains secret to the receiver, the sender in each iteration completely knows the state to be prepared. Here we keep $\theta$ to be arbitrary, i.e. the remote state may be from any circle on the surface of the Bloch sphere.

\begin{figure}[!ht]
\centering
\includegraphics[width=0.75\linewidth]{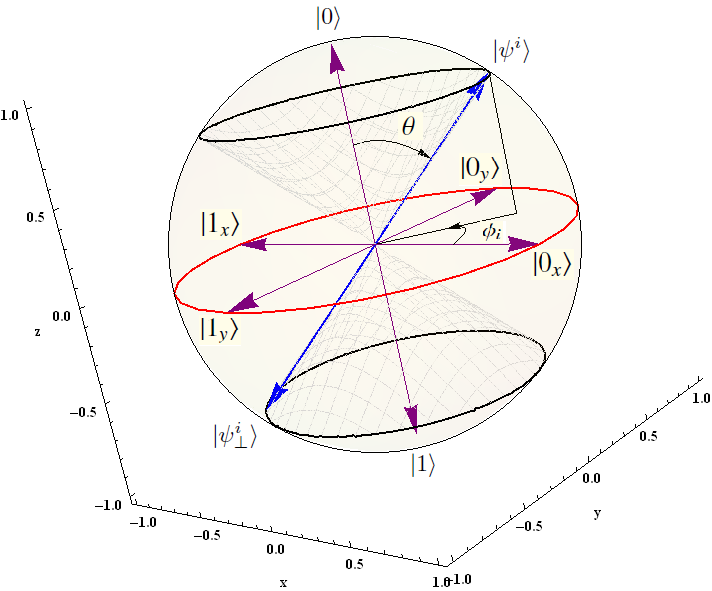}
\caption{\footnotesize (Color Online) Schematic diagram of a Bloch sphere. $\lbrace |0\rangle,|1\rangle\rbrace$ forms the computational basis, whereas $\lbrace |0_{x(y)}\rangle,|1_{x(y)}\rangle\rbrace$ are the basis corresponding to x(y) direction. $\lbrace|\psi^i\rangle,|\psi^i_{\bot}\rangle\rbrace$ are two mutually orthogonal states with a polar angle $\theta$. Traversing over all azimuthal angles $\phi_i \in [0,2\pi]$ implies, in effect, the precession over the curved surface of the double cone.}
\label{Bloch}
\end{figure}

As shown in Fig.\ref{Bloch}, for a particular azimuthal angle $\phi_i$, the remote states $|\psi^i\rangle$ and $|\psi^i_{\bot}\rangle$ are diametrically opposite to each other, i.e., they lie on a particular equatorial plane. But when we fix the polar angle, $\theta_{\neq \frac{\pi}{2}}$ and evaluate the average over all $\phi_i$s,  the vector $|\psi^i\rangle$ or its complement precess over the curved surface of the right-circular cone thus formed with its apex pointing at the center of the Bloch sphere. Choosing the remote states for a fixed $\theta$ ($0\leq \theta \leq \pi$) indicates the choice of a circle from the Bloch sphere. For all Bobs, $\theta$ is fixed a priori for a given circle of the Bloch sphere and it is known to all, including Alice through a public channel. As per the formalism, Alice has no access to $\phi_i$s, i.e., single infinity bits of information about the state to be prepared, remains unknown to her.


\begin{lemma}
For every step of remote qubit preparation at Alice's side, the average state remains discordant for the subsequent Bob.
\end{lemma}

\begin{proof}
From Bob$^2$ onwards, the pre-measurement state between Alice and Bob$^i$ can be obtained from Eq.(\ref{gs}) by using singlet state as the initial state and by averaging over all the input states from a given circle of the Bloch sphere with fixed $\theta$ and measurement outcomes chosen by the previous Bob. This becomes a Bell-diagonal state of the form
\begin{equation}
\rho^i = \frac{1}{4} \Big(\mathbb{I}_2 \otimes \mathbb{I}_2 + \sum_{j=1}^3 c_{ij} ~\sigma_j \otimes \sigma_j \Big), ~~~~~(i\geq 2)
\label{BDState}
\end{equation}
where, $\lbrace \sigma_j \rbrace_{j=1}^3$ are Pauli spin matrices and the co-efficients ($0\leq |c_{ij}| \leq 1 ~\forall i,j$) turn out to be,
\begin{align}
&c_{i1} = c_{i2} = - \prod_{k=1}^{i-1} [\sqrt{1-\lambda_k^2} + \frac{1}{2} (1-\sqrt{1-\lambda_k^2}) \sin^2 \theta], \nonumber\\
&c_{i3} = - \prod_{k=1}^{i-1} [1- (1-\sqrt{1-\lambda_k^2}) \sin^2 \theta], ~~~(i\geq 2).
\label{BDS}
\end{align}
Note that using singlet state $\rho^1$ and unsharp measurements $E_a^{\lambda_1} \forall a$, the form of $\rho^2$ can be obtained from Eq.(\ref{gs}). Now $E_a^{\lambda_1}$ depends on either $|\psi^i\rangle$ or $|\psi^i_{\bot}\rangle$ which is a function of azimuthal angle $\phi_i$ by using a fixed polar angle $\theta$. Then $\rho^2$ takes the form of a Bell-diagonal state given by Eq.(\ref{BDState}) by using $i=2$. Similarly, using $\rho^2$ and unsharp measurements $E_a^{\lambda_2} \forall a$, Eq.(\ref{gs}) can be further employed to obtain $\rho^3$, which has the form of Eq.(\ref{BDState}) with $i=3$. Thus Eq.(\ref{BDState}) is a generalised form of shared state between Alice and subsequent Bobs having the state co-efficients given by Eq.(\ref{BDS}).

The correlation matrix $M^i$, which is constructed by the elements $\lbrace M^i_{pq} | M^i_{pq} = \operatorname{Tr} [(\sigma_p \otimes \sigma_q)\rho^i]\rbrace_{p,q=1}^3$, has eigenvalues $\lbrace c_{i1}, c_{i2}, c_{i3} \rbrace$. It is easy to check that two of the eigenvalues, $c_{i1}=c_{i2}\neq 0 ~\forall \lambda_k \in [0,1], ~(k=1,...,i)$. The geometric quantum discord of $\rho^i$ can be calculated by the minimum trace distance from the set of zero-discord classical states ($\eta$) \cite{Discord,Discord1} as follows:
\begin{align}
\mathbb{D}^{(2)} (\rho^i) =& 2 \min_{\eta} \parallel \rho^i - \eta \parallel^2 ~= 2 \min_{\eta} \operatorname{Tr} (\rho^i - \eta)^2 \nonumber\\
=& \frac{1}{2} (c_{i1}^2 + c_{i2}^2 + c_{i3}^2 - \max \lbrace c_{i1}^2, c_{i2}^2, c_{i3}^2 \rbrace).
\end{align}
Hence, the geometric discord of the state shared between Alice and Bob$^i$ remains always non-zero for all $i$. This implies that the average state remains resourceful for the next round of RSP.
\end{proof}

Note that, as the joint state $\rho^i$, shared between Alice and Bob$^i$ ($i\geq 2$), belongs to the family of Bell-diagonal states with maximally mixed marginals~\cite{Horodecki96,Luo}, Bob$^i$ performs an unsharp measurement on his subsystem in the the basis $\lbrace |\psi^i\rangle, |\psi^i_{\bot}\rangle \rbrace$ and sends the result to Alice. Depending upon the outcome, the normalised conditional state at Alice's side becomes either $\rho^i_{A|E_+^{\lambda_i}}$ or $\rho^i_{A|E_-^{\lambda_i}}$ by using Eq.(\ref{cond}). We find that,
\begin{align}
\rho^i_{A|E_{\pm}^{\lambda_i}} = \frac{1}{2} [(1\pm a_i) |0\rangle\langle 0| \pm b_i |0\rangle\langle 1| \pm b_i^* |1\rangle\langle 0| + (1\mp a_i) |1\rangle\langle 1|]
\label{BDCond}
\end{align} 
where, \begin{align*}
&a_i = \lambda_i ~\cos\theta ~c_{i3}, \\
&b_i = \lambda_i ~e^{-i\phi_i} ~\sin\theta ~c_{i1} = \lambda_i ~e^{-i\phi_i} ~\sin\theta ~c_{i2} 
\end{align*}
and $b_i^*$ is the complex conjugate of $b_i$. Whenever Bob$^i$ wants to prepare $|\psi^i\rangle$ or $|\psi^i_{\bot}\rangle$ chosen from a given circle with polar angle $\theta$ on the Bloch sphere, then the average fidelity between the prepared and the desired state becomes $f_{av}= \frac{1}{2\pi} \int_0^{2\pi} \langle \psi^i | \rho^i_{A|E_-^{\lambda_i}} | \psi^i \rangle ~d\phi_i=\frac{1}{2\pi} \int_0^{2\pi} \langle \psi^i_{\bot} | \rho^i_{A|E_+^{\lambda_i}} | \psi^i_{\bot} \rangle ~d\phi_i = \frac{1}{2} - \frac{\lambda_i}{4} [(c_{i1} + c_{i2}) \sin^2 \theta + 2c_{i3} \cos^2 \theta]$. It corresponds to the 50\% successful RSP when $\theta \in (0,\frac{\pi}{2})\cup (\frac{\pi}{2},\pi)$. Whereas for completely successful RSP with $\theta=\frac{\pi}{2}$, the fidelity, averaged over all input states and measurement outcomes, becomes 
\begin{align}
f_{av}^{AB^i} =& \frac{1}{2\pi} \int_0^{2\pi} (p_{+}^i \langle \psi^i | \sigma_z.\rho^i_{A|E_+^{\lambda_i}}.\sigma_z | \psi^i \rangle + p_{-}^i \langle \psi^i | \rho^i_{A|E_-^{\lambda_i}} | \psi^i \rangle) ~d\phi_i \nonumber\\
=& \frac{1}{2\pi} \int_0^{2\pi} \langle \psi^i | \rho^i_{A|E_-^{\lambda_i}} | \psi^i \rangle ~d\phi_i \nonumber\\
=& \frac{1}{2} - \frac{\lambda_i (c_{i1} + c_{i2})}{4},
\label{fidBDS}
\end{align}
where $p_{+}^i = p_{-}^i = \frac{1}{2} ~\forall c_{ij},\phi_i,\lambda_i$. Now $\rho^i$ reduces to the well-known Werner state~\cite{Werner} with $c_{i1}=c_{i2}=c_{i3}=-c$, i.e.,
\begin{equation}
\rho_W = c |\psi^-\rangle\langle\psi^-| + \frac{1-c}{4} (\mathbb{I}_2 \otimes \mathbb{I}_2),
\label{Werner}
\end{equation}
where $0\leq c \leq 1$. Using $\rho_W$ and projective (sharp) measurements (i.e. $\lambda_i=1$) at Bob's side, we have the average RSP-fidelity as
\begin{equation}
f_{av}(\rho_W) = \frac{1+c}{2}, 
\end{equation}
which yields the quantum advantage at its best when $\theta=\frac{\pi}{2}$. Hence the non-classicality of RSP-fidelity corresponding to input states from the equatorial plane of the Bloch sphere is achieved for Werner state when $c>\frac{1}{2}$. The non-classical advantage is realised by comparing the results with the classical strategy in a particular scenario where the RSP-protocol is described. Whenever the shared state between Alice and Bob is discordant with maximally mixed marginals, it has the average RSP-fidelity $\frac{1}{2} < f_{av} \leq 1$. Fidelity is $1$ for the maximally entangled or Bell states, whereas it goes to $\frac{1}{2}$ iff the shared state is maximally mixed. In fact, RSP-fidelity goes down as the mixedness of the shared state increases as mentioned in~\cite{Discord}.

\subsection{Remote states chosen from the equatorial great circle of the Bloch sphere}

Here the polar angle of the remote states are fixed at $\theta=\frac{\pi}{2}$ on the Bloch sphere and this is known to Alice and all the subsequent Bobs.

\begin{theorem}
At most 6 number of Bobs can share the task of preparing remote state chosen from the equatorial circle of the Bloch sphere at Alice's side with non-classical advantage achieved by all of them when a maximally entangled state is initially shared between Alice and Bob$^1$.
\end{theorem}

\begin{proof}
Suppose that, Bob$^i$ wishes to prepare remote state and he communicates the result of his unsharp measurements $E_{\pm}^{\lambda_i}$ to Alice . By applying Eq.(\ref{BDCond}), the conditional state prepared at Alice's side from the states $\rho^i$ as given by Eq.(\ref{BDState}) with $\theta=\frac{\pi}{2}$, becomes either corresponding to the outcome (+) as,
\begin{align}
\rho^i_{A|E_+^{\lambda_i}} =& \frac{\lambda_i}{2^{i-1}} ~\prod_{k=1}^{i-1} \Big(1+\sqrt{1-\lambda_k^2}\Big) ~|\psi^i_{\bot}\rangle\langle\psi^i_{\bot}| \nonumber\\
&+ \frac{1}{2} \Big[ 1- \frac{\lambda_i}{2^{i-1}} ~\prod_{k=1}^{i-1} \Big(1+\sqrt{1-\lambda_k^2}\Big) \Big] ~\mathbb{I}_2, ~~~~~(i\geq 2)
\label{cond1}
\end{align}
or corresponding to the outcome (-) as,
\begin{align}
\rho^i_{A|E_-^{\lambda_i}} =& \frac{\lambda_i}{2^{i-1}} ~\prod_{k=1}^{i-1} \Big(1+\sqrt{1-\lambda_k^2}\Big) ~|\psi^i\rangle\langle\psi^i| \nonumber\\
&+ \frac{1}{2} \Big[ 1- \frac{\lambda_i}{2^{i-1}} ~\prod_{k=1}^{i-1} \Big(1+\sqrt{1-\lambda_k^2}\Big) \Big] ~\mathbb{I}_2, ~~~~~(i\geq 2).
\label{cond2}
\end{align}

Finally depending upon the CC made by Bob$^i$, Alice applies local unitary operation chosen from the set $\lbrace \mathbb{I}_2, \sigma_z \rbrace$, which is eventually the optimal set of unitaries, giving rise to maximum RSP-fidelity\footnote{Let us consider that, Alice applies a general 2$\times$2 unitary matrix having the form, $U=\begin{pmatrix} \cos(\frac{\zeta}{2}) ~e^{i(\frac{\iota+\kappa}{2})} & \sin(\frac{\zeta}{2}) ~e^{-i(\frac{\iota-\kappa}{2})}\\
-\sin(\frac{\zeta}{2}) ~e^{i(\frac{\iota-\kappa}{2})} & \cos(\frac{\zeta}{2}) ~e^{-i(\frac{\iota+\kappa}{2})}\\
\end{pmatrix}$ next to her prefixed chosen set of unitaries $\lbrace \mathbb{I}_2, \sigma_z \rbrace$. Hence the average fidelity in the given scenario between Alice and Bob$^i$ ($i\geq 2$) becomes, $f_{av}^{AB^i} \big(\rho_{A|E_{+(-)}^{\lambda_i}}^i, |\psi^i_{(\bot)}\rangle \big) = \frac{1}{2} + \frac{\lambda_i}{2^i} ~\prod_{k=1}^{i-1} \big(1+\sqrt{1-\lambda_k^2}\big) ~[\cos^2(\frac{\zeta}{2}) ~\cos(\iota+\kappa) - \sin^2(\frac{\zeta}{2}) ~\cos(\iota-\kappa-2\phi_i)]$, which is maximum when $\zeta=\iota=\kappa= 0 ~\text{or} ~2\pi$, i.e., $U=\mathbb{I}_2$.}. Therefore, the average RSP-fidelity  derived from Eq.(\ref{psi100}) or Eq.(\ref{psi_perp100}) by considering all input states and measurement outcomes  with the help of Eq.(\ref{cond1}) and Eq.(\ref{cond2}) becomes (see Appendix \ref{A})
\begin{align}
f_{av}^{AB^i} = \frac{1}{2} + \frac{\lambda_i}{2^i} \prod_{k=1}^{i-1} \Big(1+\sqrt{1-\lambda_k^2}\Big), ~~~~~(i\geq 2).
\label{avfid}
\end{align}
Eq.(\ref{avfid}) can be reproduced by plugging Eq.(\ref{BDS}) into Eq.(\ref{fidBDS}) with $\theta=\frac{\pi}{2}$. Note that Eq.(\ref{avfid}) is obtained without a need for the post-selection subject to the measurement done by Bob$^i$ as $|\psi^i_{\bot}\rangle$ can be transformed into $|\psi^i\rangle$ under the action of $\sigma_z$ for the remote states from the equatorial circle of the Bloch sphere.

The success of Bob$^i$ is indicated by $f_{av}^{AB^i} > \frac{3}{4}$ compared to the classical strategy, and there can be found a range of $\lambda_i, (i\geq 1)$ in each iteration ($i$) for which such fidelity can be achieved. Using the range of $\lambda_j, (j<i)$ corresponding to all the previous Bobs upto Bob$^j$ $\forall j \in \lbrace 1,2,...,i-1 \rbrace$, being capable of performing RSP at Alice's side, a new range of $\lambda_i$ can be specified  for the i-th Bob. Proceeding in this way,  we obtain the particular ranges of sharpness parameters given by the second column of Table \ref{tab:t1}, for which quantum fidelity of RSP is achieved up to the corresponding Bobs. 

As the number of Bobs increases, the last Bob in the sequence has to perform sharper measurements in order to prepare remote state at Alice's side with average fidelity $> \frac{3}{4}$.  The condition to achieve the task of RSP for Bob$^i$, depending upon the ability of all the previous Bobs upto Bob$^{i-1}$ to do the same task in their turn, determines the minimum sharpness of measurement for Bob$^i$. For example, Bob$^1$ is able to accomplish the task of RSP when $\frac{1}{2} < \lambda_1 \leq 1$. Next, Bob$^2$ prepares remote states from the equatorial great circle of the Bloch sphere with average RSP-fidelity given by Eq.(\ref{avfid}) with $i=2$. Now, the occurrence of $f_{av}^{AB^2} > \frac{3}{4}$ implies that, $\frac{1}{1+\sqrt{1-\lambda_1^2}} < \lambda_2 \leq 1$ under the restriction of Bob$^1$'s ability to perform the task of RSP, i.e., within the limit $\frac{1}{2} < \lambda_1 \leq 1$. Hence, the numerical minimum of $\lambda_2$ occurs when $\lambda_1$ tends to $\frac{1}{2}$, i.e., $\lim_{\lambda_1 \rightarrow \frac{1}{2}^+} \frac{1}{1+\sqrt{1-\lambda_1^2}} = 0.536$. In this way, the minimum sharpness of measurements for subsequent Bobs to perform RSP are calculated and the range of sharpness parameters are shown in the second column of Table \ref{tab:t1}.

\begin{table}[h!]
\centering
\begin{adjustbox}{width=0.6\textwidth}
 \begin{tabular}{| c | c | c |} 
 \hline
 $i ~\hat{=}$ Bob$^i$ & Range of $\lambda_i$ & $\max_{\lbrace \lambda_k \rbrace_{k=1}^i} \frac{\sum_{k=1}^i f_{av}^{AB^k}}{\sum_{k=1}^i k}$ \\ 
 \hline
 1 & (0.5 - 1] & 1 \\ 
 \hline
 2 & (0.536 - 1] & 0.904 \\
 \hline
 3 & (0.581 - 1] & 0.849 \\
 \hline
 4 & (0.641 - 1] & 0.812 \\
 \hline
 5 & (0.725 - 1] & 0.785 \\  
 \hline
 6 & (0.859 - 1] & 0.764 \\
 \hline
\end{tabular}
\end{adjustbox}
\caption{The domain of sharpness parameters for which 6 Bobs  can attain average RSP-fidelity $> \frac{3}{4}$ by using Eq.(\ref{avfid}). The maximum  average RSP-fidelity at each round of the protocol is shown in the third column where it can be seen that  our optimal classical bound, i.e., $\frac{3}{4}$ can be violated by a significant amount.}
\label{tab:t1}
\end{table}

It can be checked that after the 6-th Bob accomplishing the task of RSP, even a sharp measurement by Bob$^{7}$ can achieve a maximum of 0.72 as the value of average RSP-fidelity. This shows that it is not possible for Bob$^{7}$ onwards to execute the task with the required quantum fidelity. 

In the third column of Table \ref{tab:t1} we display the maximum average RSP-fidelity ($\max_{\lbrace \lambda_k \rbrace_{k=1}^i} \frac{\sum_{k=1}^i f_{av}^{AB^k}}{\sum_{k=1}^i k}$) that can be achieved by the corresponding number of Bobs at every stage of the RSP-protocol.  By considering multiple Bobs up to Bob$^6$, Bobs on average can at most achieve the given values of average RSP-fidelities which are well above the classical limit $\frac{3}{4}$. On the other hand, by applying the same technique for all Bobs up to Bob$^7$, the average of RSP-fidelity can attain a maximum of $0.747$, subject to the sharp measurement done by Bob$^7$, which is less than our classical limit of RSP-fidelity. We assume that the same technique is followed by all the Bobs, so that no Bob is allowed to change the strategy to a classical or hybrid strategy instead of the quantum strategy, at will. The task of RSP can not be processed further when there is no way to achieve non-classical advantage by the last Bob in the sequence. Hence, at most 6 Bobs are able to sequentially manifest the task of RSP under any circumstances in the given framework where all of them achieve success against our optimal classical strategy.
\end{proof}

\subsubsection*{Geometric Discord vs Concurrence}


The quantification of resources is important to characterize the efficacy of a protocol, which may be continued further by utilizing the resources in the subsequent steps. Here we consider geometric quantum discord($\mathbb{D}^{(2)}$) and concurrence($\mathbb{C}$) as the resources for the subsequent Bobs participating in the the RSP-protocol. Corresponding to Bob$^i$, we have $c_{i1}^2 = c_{i2}^2 \geq c_{i3}^2$ for all $i\geq 2$ by using $\theta=\frac{\pi}{2}$ in Eq.(\ref{BDS}). Hence 
\begin{equation}
\mathbb{D}^{(2)} (\rho^{i_{\geq 2}}) = \frac{1}{2} \Big(c_{i1}^2|_{(\theta=\frac{\pi}{2})} + c_{i3}^2|_{(\theta=\frac{\pi}{2})}\Big),
\end{equation}
where $c_{i1}|_{(\theta=\frac{\pi}{2})} = c_{i2}|_{(\theta=\frac{\pi}{2})} = -\frac{1}{2^{i-1}} ~\prod_{k=1}^{i-1} (1+\sqrt{1-\lambda_k^2})$ and $c_{i3}|_{(\theta=\frac{\pi}{2})} = - \prod_{k=1}^{i-1} \sqrt{1-\lambda_k^2}$. For Bob$^1$, $\mathbb{D}^{(2)} (|\psi^-\rangle\langle\psi^-|) = 1$. The maximum utilizable resource for Bob$^{i}$ ($i\geq 2$) to perform the task of RSP becomes, $\max_{\lbrace\lambda_k\rbrace_{k=1}^{i-1}} \mathbb{D}^{(2)} (\rho^{i})$.

On the other hand, non-zero measure of concurrence implies a bipartite qubit state to be entangled~\cite{Hill}. Concurrence is maximum for a Bell state, e.g. $\mathbb{C}(|\psi^-\rangle\langle\psi^-|)=1$, whereas for a Bell diagonal state $\rho^i$ ($i\geq 2$), we can construct a matrix $\mathbb{R}=\sqrt{\sqrt{\rho^i} ~(\rho^i)^* \sqrt{\rho^i}}$ and $\rho^i$ has the same eigenspectrum as that of $\mathbb{R}$. Let us call the largest eigenvalue of $\rho^i$ as $\tau^i$. Therefore, the concurrence function~\cite{Quan} can be expressed as,
\begin{equation}
\mathbb{C}(\rho^i) = \max \lbrace 0, 2\tau^i -1 \rbrace
\end{equation}
where we have 
\begin{equation}
\tau^i = \frac{1}{4} \Big[ 1+ \prod_{k=1}^{i-1} \sqrt{1-\lambda_k^2} + \frac{1}{2^{i-2}} \prod_{k=1}^{i-1} \Big( 1 + \sqrt{1-\lambda_k^2} \Big)\Big], ~~~~~(i\geq 2)
\end{equation}
corresponding to $\theta=\frac{\pi}{2}$. The maximum resource in terms of entanglement between Alice and Bob$^i$ ($i\geq 2$) becomes, $\max_{\lbrace\lambda_k\rbrace_{k=1}^{i-1}} \mathbb{C}(\rho^{i})$. The maximization of the resources, whether it is geometric discord or entanglement w.r.t. the $\lambda_k$ values, are chosen such that, for all the previous Bobs the state remains useful for RSP outside the classical domain characterized by our classical strategy. In other words, the maximum resource remaining for Bob$^i$ is such that, all Bobs from Bob$^1$ to Bob$^{i-1}$ are successful in the task of RSP. We emphasize that this is not equivalent to sharing of resource in terms of geometric discord or entanglement among multiple observers at one side. 

We compute the maximum available resources in terms of both geometric discord and concurrence that  remain after every successful step of RSP against our classical strategy. The resource gets diminished after subsequent measurements of Bobs. In each case we can calculate the maximum utilizable resource present in the shared state to know whether the state can be further used for RSP or not. For instance, after measurements by Bobs successful for RSP upto Bob$^{i-1}$, the maximum discord that remains in the state $\rho^i$ shared between Alice and Bob$^i$ can be calculated such that Bob$^i$'s ability to perform RSP by using $\rho^i$ can be known. The maximization depends on the ability of all Bobs upto Bob$^{i-1}$ to perform RSP specified by the range of sharpness parameters shown in the column 2 of Table \ref{tab:t1} where the remote states are prepared from the equatorial great circle of the Bloch sphere. Similarly, the concurrence as a measure of entanglement for the states $\rho^i$ are maximized. 

\begin{figure}[!ht]
\centering
\includegraphics[width=0.75\linewidth]{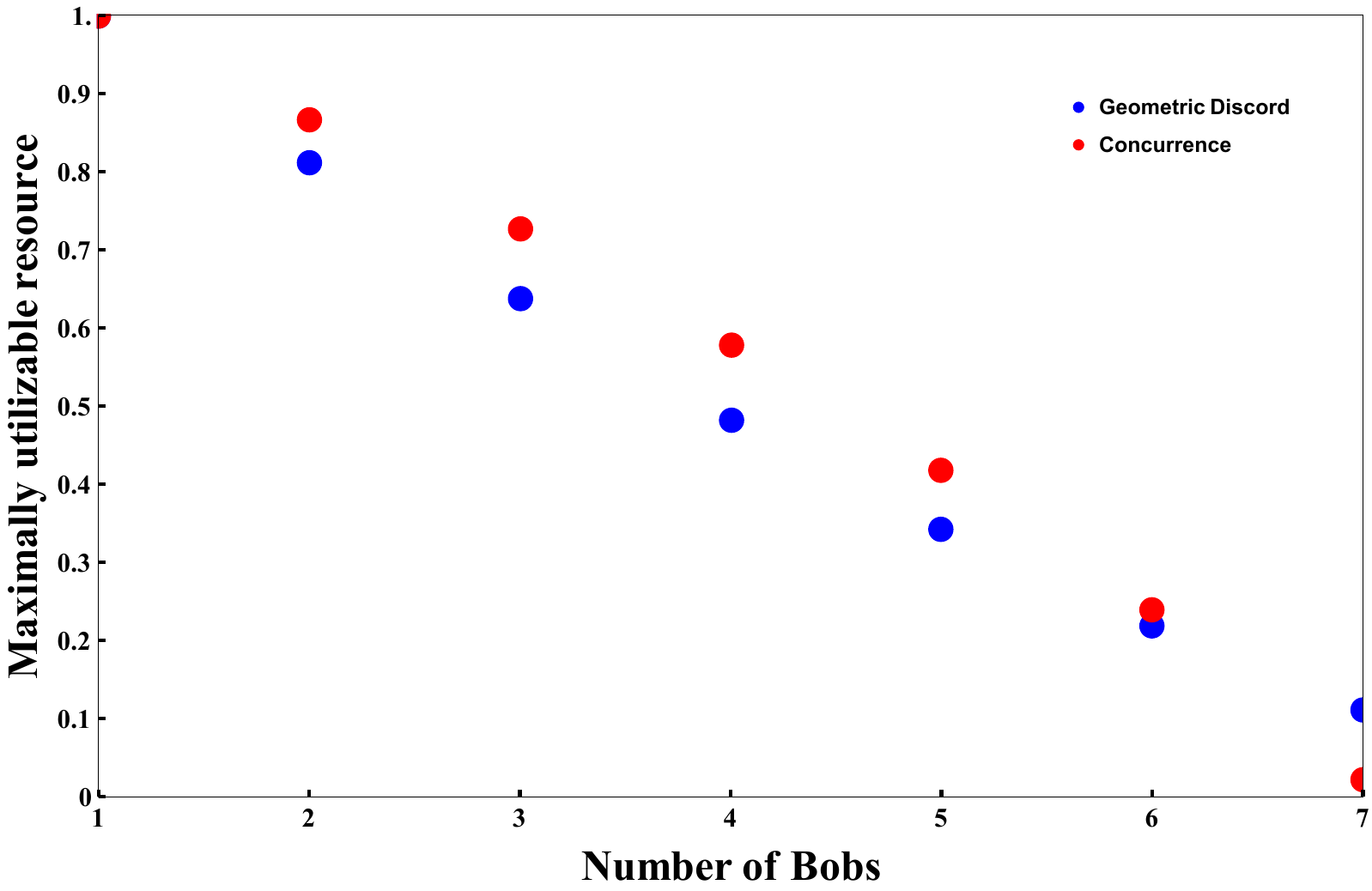}
\caption{\footnotesize (Color Online) The maximum geometric discord and concurrence available for different number of Bobs are presented whenever the RSP, with fidelity higher than that predicted by our classical strategy, is possible upto its previous Bobs.}
\label{resource}
\end{figure}

We can thus plot in Fig.\ref{resource} the geometric discord and concurrence which the Bobs can maximally utilize in the successive steps to perform RSP successfully beyond the classical region delimited by our classical strategy. In comparison with this, we observe that even after the success of Bob$^{6}$ in the task of RSP, both the geometric quantum discord and the concurrence remain positive which indicates the possibility for the 7-th Bob to perform RSP as well. From Table \ref{tab:t1}, it can be easily anticipated that, $\lambda_7$ is no less than $0.859$ and for such value of $\lambda_7$, the concurrence vanishes while the geometric discord still remains positive. However, in contrast to such possibility it is shown earlier that sequential success in the RSP-protocol by all the Bobs beyond Bob$^6$ can not be realised with non-classical advantage.

It may be noted that the average fidelity of the i-th state produced by Bob$^i$ at Alice's side reduces with the sharpness of the i-th measurement. Here the minimum of $\lambda_i \forall i$ plays an important role to obtain the maximum number of possible Bobs in the given framework, which we aim to figure out in different contexts. If one considers sharpness parameters higher than the minimum value, then  the number of Bobs sharing the task of RSP will be  lower than the maximum number of Bobs derived in our work. The range of sharpness parameters, i.e., $(\lambda_i)_{\min} < \lambda_i \leq 1$ are defined corresponding to the range of average RSP-fidelity $\frac{3}{4} < f_{av}^{AB^i} \leq 1$. The minimum of sharpness parameters, $(\lambda_i)_{\min}$ correspond to $f_{av}^{AB^i}$ tending to $\frac{3}{4}$ which implies a negligible amount of quantum advantage in the RSP-protocol and becomes difficult to test in practice. Hence, by fixing $f_{av}^{AB^i}$ at a higher value, the minimum of $\lambda_i$ will be higher, and a lesser number of Bobs will be able to accomplish RSP. Nonetheless, if we fix the number of sequential Bobs to be 6, then the Bobs  can optimally achieve the average RSP-fidelity of 0.764, as presented in the third column of Table \ref{tab:t1}, where the quantum advantage is not negligible. The experimental viability of such a protocol may become clearer with further robustness analysis.  On the other hand, even a negligible quantum superiority can not be attained on average by fixing 7 number of Bobs in a sequence and by applying the same approach to calculate the average RSP-fidelity for all of them.

It may be pertinent to note here that since the bound on the number of Bobs depends on the classical strategy with respect to  which the quantum advantage is manifested, various other contexts may be studied to realize the bound under dissimilar frameworks. For example, when the classical bound on the RSP-fidelity is defined based on a different strategy  under special circumstances~\cite{Horodecki14},  by comparing it with the quantum strategy, it may be observed that at most 3 sequential Bobs can achieve success in the task of preparing remote states from the equatorial circle of the Bloch sphere (see Appendix \ref{AA}).

\subsection{Remote states chosen from a non-equatorial circle of the Bloch sphere}

Let us suppose that the remote state to be prepared at Alice's side is  chosen from any non-equatorial circle with fixed $\theta$ (i.e. not equal to $\frac{\pi}{2}$) on the Bloch sphere. This is known to all the Bobs and Alice. Here Alice applies suitable post-selection technique corresponding to the CC received from Bobs in each iteration as the RSP-protocol is 50\% successful in every case due to non-feasibility of the conversion between $|\psi^i\rangle$ and $|\psi_{\bot}^i\rangle ~\forall i$.

\begin{theorem}
The maximum number of Bobs sharing the task of RSP at Alice's side from a non-equatorial circle of the Bloch sphere with $\theta \neq \frac{\pi}{2}$, is less than or equal to 6 with non-classical advantage achieved by all of them when Alice shares a maximally entangled state with Bob$^1$ initially. As the choice of the circle moves from the neighbourhood of the equatorial plane towards the poles of the Bloch sphere in either direction, the maximum number of Bobs reduces gradually to zero.
\end{theorem}

\begin{proof}
For a randomly chosen $\theta$, the average RSP-fidelity at the $i$-th step can be determined by using Eq.(\ref{BDState}) and Eq.(\ref{BDCond}) subject to unsharp measurement done by Bob$^i$ with a precondition that it is no less than $f_{cl}^{\max}$ corresponding to the up outcome at Bob$^i$'s side as Alice post-selects the state to consider for RSP corresponding to the down outcome at Bob$^i$'s side. This method is applied invariably for all the Bobs. Therefore the average RSP-fidelity by applying Eq.(\ref{psi50}) or Eq.(\ref{psi_perp50}) eventually becomes a function of all $\lambda_i$s ($i\geq 2$) and $\theta$, as well. The average fidelity of the conditional state at Alice's side for preparing remote states $\lbrace |\psi^i\rangle \rbrace_{\theta \neq \frac{\pi}{2},\phi_i}$ from a circle with fixed polar angle $\theta\neq \frac{\pi}{2}$ on the Bloch sphere, contingent upon the possible choices and outcomes of the unsharp measurement done by Bob$^i$, becomes
\begin{align}
f_{av}^{AB^i} =& \frac{1}{2\pi} \int_0^{2\pi} (p_{+}^i ~f_{cl}^{\max} + p_{-}^i ~\langle \psi^i | \rho^i_{A|E_{-}^{\lambda_i}} | \psi^i \rangle) ~d\phi_i \nonumber\\
=& \frac{f_{cl}^{\max}}{2} + \frac{1}{4} + \frac{\lambda_i}{4} \Big[ \cos^2\theta ~\prod_{k=1}^{i-1} \Big( \cos^2\theta + \sin^2\theta \sqrt{1-\lambda_k^2} \Big) \nonumber\\
&+ \frac{\sin^2\theta}{2^{i-1}} ~\prod_{k=1}^{i-1} \Big( \sin^2\theta + (\cos^2\theta +1) \sqrt{1-\lambda_k^2} \Big) \Big], ~~~(i\geq 2)
\label{Noneq}
\end{align}
where $p_{+}^i = p_{-}^i = \frac{1}{2} ~\forall \theta,\phi_i,\lambda_i ~(i\geq 1)$. 

Similar expressions for the preparation of remote states $\lbrace |\psi^i_{\bot}\rangle \rbrace_{\theta \neq \frac{\pi}{2},\phi_i}$ can also be found with $f_{cl}^{\max}$ as the optimal fidelity corresponding to the down outcome at Bob$^i$'s side, i.e., $f_{av}^{AB^i} = \frac{1}{2\pi} \int_0^{2\pi} (p_{+}^i ~\langle \psi^i_{\bot} | \rho^i_{A|E_{+}^{\lambda_i}} | \psi^i_{\bot} \rangle + p_{-}^i ~f_{cl}^{\max}) ~d\phi_i$. In this case, Alice rejects the state corresponding to the down outcome communicated by Bob$^i$ corresponding to the $i$-th iteration. For example, the fidelity between Alice and Bob$^2$ in the 2nd iteration turns out as, $f_{av}^{AB^2} = \frac{f_{cl}^{\max}}{2} + \frac{1}{4} + \frac{\lambda_2}{64} [(9+7\sqrt{1-\lambda_1^2})+(1-\sqrt{1-\lambda_1^2}) (4 \cos 2\theta + 3 \cos 4\theta)]$. Now $f_{av}^{AB^2} > f_{cl}^{\max}$ is achieved when $\lambda_2 > \frac{16 + 12 \sin\theta + 16 \sin 2\theta - 4 \sin 3\theta}{18 + 8 \cos 2\theta + 6 \cos 4\theta + (14-8\cos 2\theta - 6 \cos 4\theta) s(\theta)}$ where $s(\theta)=\sqrt{1-\frac{(1+\cos 2\theta + \sin^3 \theta)^2}{4}}$ under the constraint $\lambda_1 > \frac{1+\cos 2\theta + \sin^3 \theta}{2}$ imposed by Bob$^1$. Similarly, we can find the minimum sharpness parameters for subsequent Bobs as functions of $\theta$.

\begin{figure}[!ht]
\centering
\includegraphics[width=0.75\linewidth]{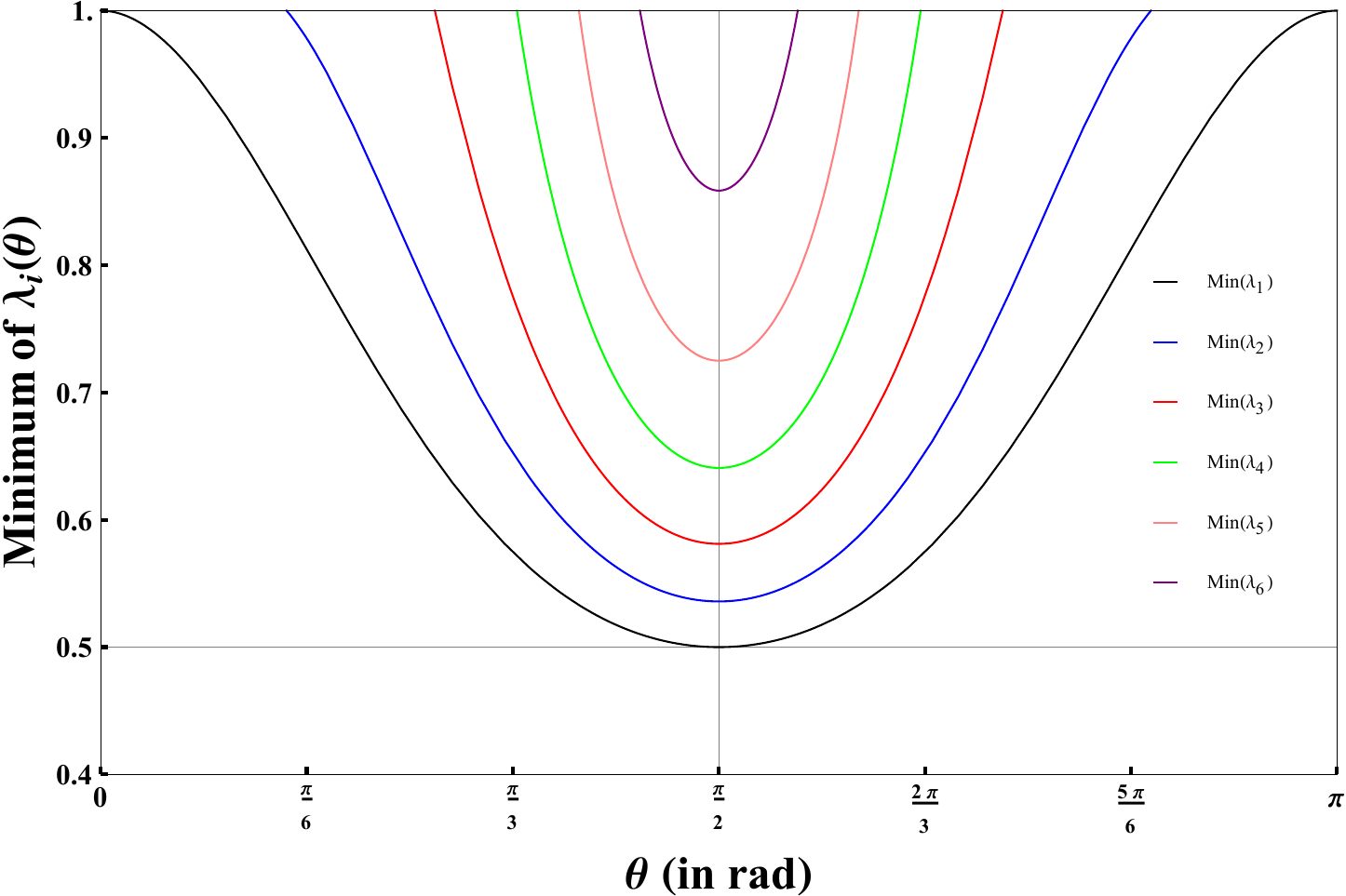}
\caption{\footnotesize (Color Online) $\lbrace(\lambda_i)_{\min}\rbrace_1^6$ are plotted against different polar angles of the remote state ($\theta$) such that, the average RSP-fidelities given by Eq.(\ref{Noneq}) can attain non-classical values well above the upper bound predicted by our classical strategy for $(\lambda_i)_{\min} < \lambda_i \leq 1 ~\forall i$, maintained with the corresponding ranges of sharpness parameters manifested by the previous Bobs (up to Bob$^{i-1}$).}
\label{rsps}
\end{figure}

We plot minimum $\lambda_i$s as  function of $\theta$ in subsequent iterations (see Fig.\ref{rsps}) to overcome the classical limit of fidelity. We observe that there can not be more than 6 number of Bobs who can sequentially prepare remote states at Alice's side from a circle with polar angle $\theta$ ($\forall \theta \in[0,\pi]$) while comparing the average RSP-fidelities with that predicted by our classical strategy. It can be seen from Fig.\ref{rsps} that, $(\lambda_i)_{\min}(\theta) \forall i$ is an even function w.r.t. $\theta=\frac{\pi}{2}$. It can also be observed that, $(\lambda_i)_{\min}(\theta)$ has minima at $\theta=\frac{\pi}{2} ~\forall i \leq 6$ in order to attain $f_{av}^{AB^i}> f_{cl}^{\max}$. Thus the highest number of Bobs ($n=6$) can be employed for sharing the task of RSP within the permissible region of $\lambda_i \in((\lambda_i)_{\min}, 1] ~\forall i \leq n$ when the remote states are prepared from the equatorial circle of the Bloch sphere ($\theta=\frac{\pi}{2}$). The permissible regions of $\lambda_i$ is displayed earlier in Table \ref{tab:t1}. 

For non-equatorial circles in the neighbourhood of $\theta=\frac{\pi}{2}$, the maximum number of Bobs remains 6, whereas it gradually becomes less than 6 as the chosen circle on the Bloch sphere progresses towards the poles of the Bloch sphere with polar angle $\theta=0$ or $\pi$ where the RSP-protocol demonstrates quantum mechanical advantage against our classical strategy in all the cases. It can be checked that $(\lambda_1)_{\min}=1$ occurs at the poles, and hence, there is no quantum advantage of RSP  even by a single Bob.  The maximum number of Bobs sharing the task of RSP gradually lowers towards the poles of the Bloch sphere because, as the perimeter of the chosen circle on the Bloch sphere shrinks, the possible number of input states becomes lower which makes the probability of guessing the remote state higher for Alice. It is evident from Fig.\ref{fidelity} that the classical fidelity limit increases to 1 as the chosen circle moves from the equatorial plane towards the poles of the Bloch sphere. Therefore it becomes more and more difficult for the subsequent Bobs to gain quantum advantage by violating the classical fidelity limit within $(\lambda_i)_{\min} < \lambda_i \leq 1$. Table\ref{tab:n2} shows the maximum number of Bobs, $n$ ($i \leq n$) sharing the task of RSP successfully against our classical strategy for a given $\theta$ such that $\lambda_{n+1}$ for Bob$^{n+1}$ exceeds the allowed range $\lambda_{n+1} < 1$. Here we apply the technique to calculate the average RSP-fidelity irrespective of all the Bobs such that a Bob can not randomly select any of the classical, quantum or mixed strategies. The task of sharing RSP discontinues when the last Bob for a given iteration becomes unable to attain the non-classical average RSP-fidelity.

\begin{table}[h!]
\centering
\begin{adjustbox}{width=0.4\textwidth}
 \begin{tabular}{| c | c |} 
 \hline
 $n$ & Range of $\theta$ \\ 
  & (in rad) \\
 \hline
 1 & $(0,0.472]\cup[2.669,\pi)$ \\ 
 \hline
 2 & $(0.472,0.849]\cup[2.292,2.669)$ \\
 \hline
 3 & $(0.849,1.058]\cup[2.084,2.292)$ \\
 \hline
 4 & $(1.058,1.215]\cup[1.926,2.084)$ \\
 \hline
 5 & $(1.215,1.370]\cup[1.771,1.926)$ \\  
 \hline
 6 & $(1.370,1.771)$ \\
 \hline
\end{tabular}
\end{adjustbox}
\caption{At most $n$-number of Bobs can sequentially prepare the remote states with polar angles $\theta$ such that every Bob attains the average RSP-fidelity$> f_{cl}^{\max}$ where the success of RSP-protocol by utilizing Eq.(\ref{Noneq}) compared to our classical strategy is implied by the columns 2.}
\label{tab:n2}
\end{table}

\begin{corollary}
No Bob is able to achieve quantum advantage when the remote state at Alice's side is picked from one of the poles of the Bloch sphere with the polar angles given by $\theta=\lbrace 0,\pi \rbrace$. 
\end{corollary}

\begin{proof}
The conditional state, produced at Alice's side depending upon the outcomes $\lbrace \pm \rbrace$ of the unsharp measurement performed by Bob$^i$, is given by, respectively,
\begin{align}
\rho^i_{A|E_+^{\lambda_i}} = \lambda_i |\psi^i_{\bot}\rangle\langle\psi^i_{\bot}| + \frac{1-\lambda_i}{2} \mathbb{I}_2, ~~~~~(i\geq 1) \end{align}
or,
\begin{align}
\rho^i_{A|E_-^{\lambda_i}} = \lambda_i |\psi^i\rangle\langle\psi^i| + \frac{1-\lambda_i}{2} \mathbb{I}_2, ~~~~~(i\geq 1).
\end{align}
Hence, the average RSP-fidelity corresponding to Bob$^i$ becomes,
\begin{equation}
f_{av}^{AB^i} = \frac{1+\lambda_i}{2}, ~~~~~(i\geq 1)
\label{pole}
\end{equation}
which is independent of all unsharp measurements performed by the previous Bobs. Eq.(\ref{pole}) can be easily derived from Eq.(\ref{Noneq}) by using $\theta=\lbrace 0,\pi \rbrace$ as the RSP-protocol is 100\% successful here with no need for the post-selection. The transformation $|0\rangle \rightarrow |1\rangle$ or vice-versa is allowed under the application of a NOT gate ($\sigma_x$)~\cite{Pati}. However, the operations, $\lbrace \mathbb{I}_2, \sigma_z\rbrace$ at Alice's side can only give rise to 50\% successful RSP-protocol as compared to the our classical strategy, where $f_{av}^{AB^i}=\frac{f_{cl}^{\max}}{2} + \frac{1+\lambda_i}{4}$. Now, $f_{cl}^{\max}|_{(\theta= 0,\pi)}=1$. Thus $f_{av}^{AB^i} \ngtr 1 ~\forall\lambda_i \in [0,1] (i \geq 1)$. Note that $\theta=0$ or $\pi$ corresponds to a specific remote state, $|\psi^i\rangle ~\forall (i\geq 1)$ to be either $|0\rangle$ or $|1\rangle$ (see Fig.\ref{Bloch}) and Alice can completely recognize the state by using the knowledge of $\theta$ beforehand without using any quantum resource. 
\end{proof}

\begin{figure}[!ht]
\centering
\includegraphics[width=0.75\linewidth]{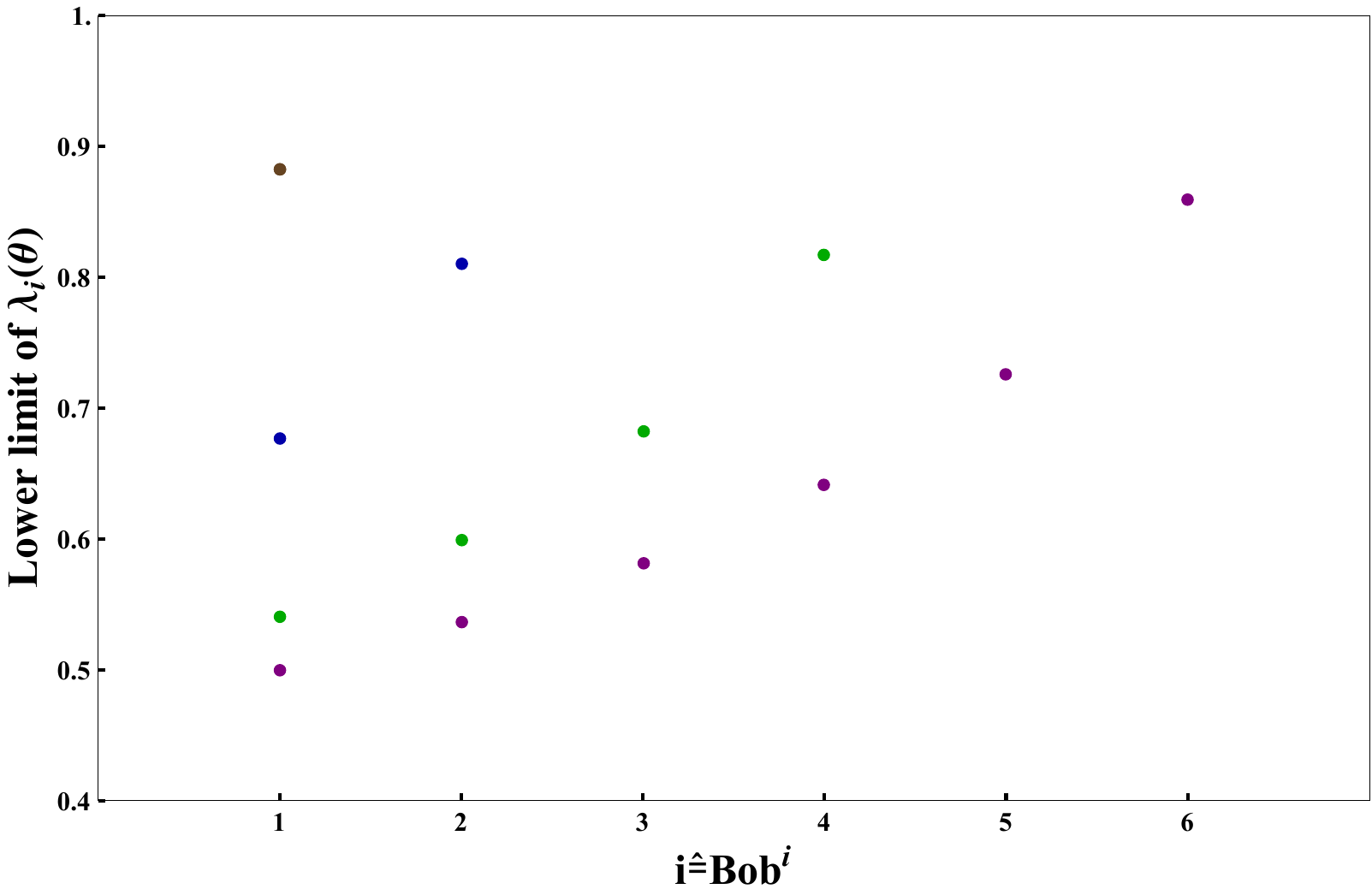}
\caption{\footnotesize (Color Online) The data points corresponding to $(\lambda_i)_{\min}(\theta)$ is plotted w.r.t. the number of Bobs ($i\leq n$) such that, for Bob$^{n+1}$ in the given scenario $(\lambda_{n+1})_{\min}>1$. Purple, green, blue, brown data points represent $\theta=\frac{\pi}{2}, \frac{3\pi}{8}, \frac{\pi}{4}$ and $\frac{\pi}{8}$ respectively. $n$ is indicated by the total number of data points for a given $\theta$.}
\label{rspsn1}
\end{figure}

We now choose a few circles of the Bloch sphere with polar angles restricted by $0<\theta < \frac{\pi}{2}$ from where the remote states are to be prepared. In Fig.\ref{rspsn1} we show that the maximum number of Bobs sharing RSP with Alice reduces with $\theta_{\neq \frac{\pi}{2}}$ compared to 6 for $\theta=\frac{\pi}{2}$ by comparing the results with our classical strategy. For instance, the bound ($n$) is 4 for $\theta=\frac{3\pi}{8}$, 2 for $\theta=\frac{\pi}{4}$ and 1 for $\theta=\frac{\pi}{8}$. We also discuss an illustration explicitly in Appendix \ref{B}.
\end{proof}

\subsection*{Sequential RSP with non-maximally  entangled state}

So far we have discussed the scenario by considering a maximally entangled initial state, which is also a maximally discordant state~\cite{Luo}. It is known that requirement of entanglement is not essential for the preparation of remote state~\cite{Discord}. Hence, it might be interesting to use a non-maximally entangled pure initial state of the form, 
\begin{equation} 
|\psi\rangle = \cos \xi |01\rangle - \sin \xi |10\rangle, ~(0\leq \xi \leq \frac{\pi}{2}).
\label{nmes} 
\end{equation}
Here $\rho^1(\xi)=|\psi\rangle\langle\psi|$. If Bob$^1$ performs unsharp measurement with sharpness parameter $\lambda_1$, the fidelity of preparing remote states from the equatorial circle of the Bloch sphere takes the form, $f_{av}^{AB^1} = \frac{1}{2} (1+\lambda_1 \sin 2\xi)$, which is maximum when $\xi=\frac{\pi}{4} ~\forall \lambda_1\in[0,1]$. Note that, when Bob$^1$ performs sharp measurement (i.e. with $\lambda_1 =1$), then $f_{av}^{AB^1} \neq 1$ except $\xi=\frac{\pi}{4}$. It implies that, non-maximally entangled state can not achieve 100\% success in the RSP-protocol due to the lack of rotational invariance in $\rho^1(\xi)$ in terms of the basis representation as given by Eq.(\ref{singlet}). However, our interest is to find the region of $\lambda_1$ for which $f_{av}^{AB^1}$ gives the non-classical advantage. And we find that, Bob$^1$ achieves non-classical RSP-fidelity by using unsharp measurement, i.e., greater than $\frac{3}{4}$ when $\lambda_1 > \frac{1}{2} \csc(2\xi)$, enabling subsequent Bobs to repeat the task of RSP.

\begin{theorem}
The maximum number of Bobs who sequentially share the task of preparing remote state at Alice's end from the equatorial circle (i.e. $\theta=\frac{\pi}{2}$) of the Bloch sphere, reduces gradually from six to zero with non-classical advantage achieved by all of them when the initially shared pure state varies from a maximally entangled state towards a pure product state.
\end{theorem}

\begin{proof}
Corresponding to the $i-th$ Bob, the average pre-measurement state takes the form of a two-qubit $X$-state whose correlation matrix has non-zero eigenvalues $\lbrace -\prod_{k=1}^{i-1} \sqrt{1-\lambda_k^2}, -\frac{1}{2^{i-1}} \prod_{k=1}^{i-1} (1+\sqrt{1-\lambda_k^2}) \sin 2\xi, -\frac{1}{2^{i-1}} \prod_{k=1}^{i-1} (1+\sqrt{1-\lambda_k^2}) \sin 2\xi \rbrace, ~(i\geq 2)$ for all $\xi \in (0,\frac{\pi}{2})$. Therefore, the state becomes resourceful for use in the subsequent iterations. By using Eq.(\ref{psi100}) or Eq.(\ref{psi_perp100}), the average fidelity of the conditional state at Alice's side, contingent upon the unsharp measurement done by Bob$^i$, becomes
\begin{align}
f_{av}^{AB^i} =& \frac{1}{2\pi} \int_0^{2\pi} (p_{+}^i ~\langle \psi^i | \sigma_z. \rho^i_{A|E_+^{\lambda_i}}. \sigma_z | \psi^i \rangle \nonumber\\
&+ p_{-}^i ~\langle \psi^i | \mathbb{I}_2. \rho^i_{A|E_-^{\lambda_i}}. \mathbb{I}_2 | \psi^i \rangle) ~d\phi_i \nonumber\\
=& \frac{1}{2\pi} \int_0^{2\pi} (p_{+}^i ~\langle \psi^i_{\bot} | \mathbb{I}_2. \rho^i_{A|E_+^{\lambda_i}}. \mathbb{I}_2 | \psi^i_{\bot} \rangle \nonumber\\
&+ p_{-}^i ~\langle \psi^i_{\bot} | \sigma_z. \rho^i_{A|E_-^{\lambda_i}}. \sigma_z | \psi^i_{\bot} \rangle) ~d\phi_i \nonumber\\
=& \frac{1}{2} + \frac{\lambda_i}{2^i} \sin 2\xi ~\prod_{k=1}^{i-1} \Big(1+\sqrt{1-\lambda_k^2}\Big), ~~~~~(i\geq 2)
\label{nms}
\end{align}
where $p_{+}^i=p_{-}^i=\frac{1}{2}~\forall \xi,\phi_i,\lambda_i$. As a result, the RSP-fidelity is maximum for the maximally entangled initial state, i.e., $\xi=\frac{\pi}{4}$ for all the cases. The quantum supremacy of the protocol against our classical strategy can be achieved when $f_{av}^{AB^i} > \frac{3}{4}$ which is analogous to obtaining $(\lambda_i)_{\min}(\xi) < \lambda_i \leq 1$ for all Bob$^i$ ($i\geq 2$). For example, Bob$^2$ can successfully prepare a remote state from the equatorial circle of the Bloch sphere at Alice's side when $\lambda_2 > 4\sin 2\xi (1-\sqrt{1-\frac{\csc^2 2\xi}{4}})$ as per the comparison with the classical strategy.

\begin{figure}[!ht]
\centering
\includegraphics[width=0.75\linewidth]{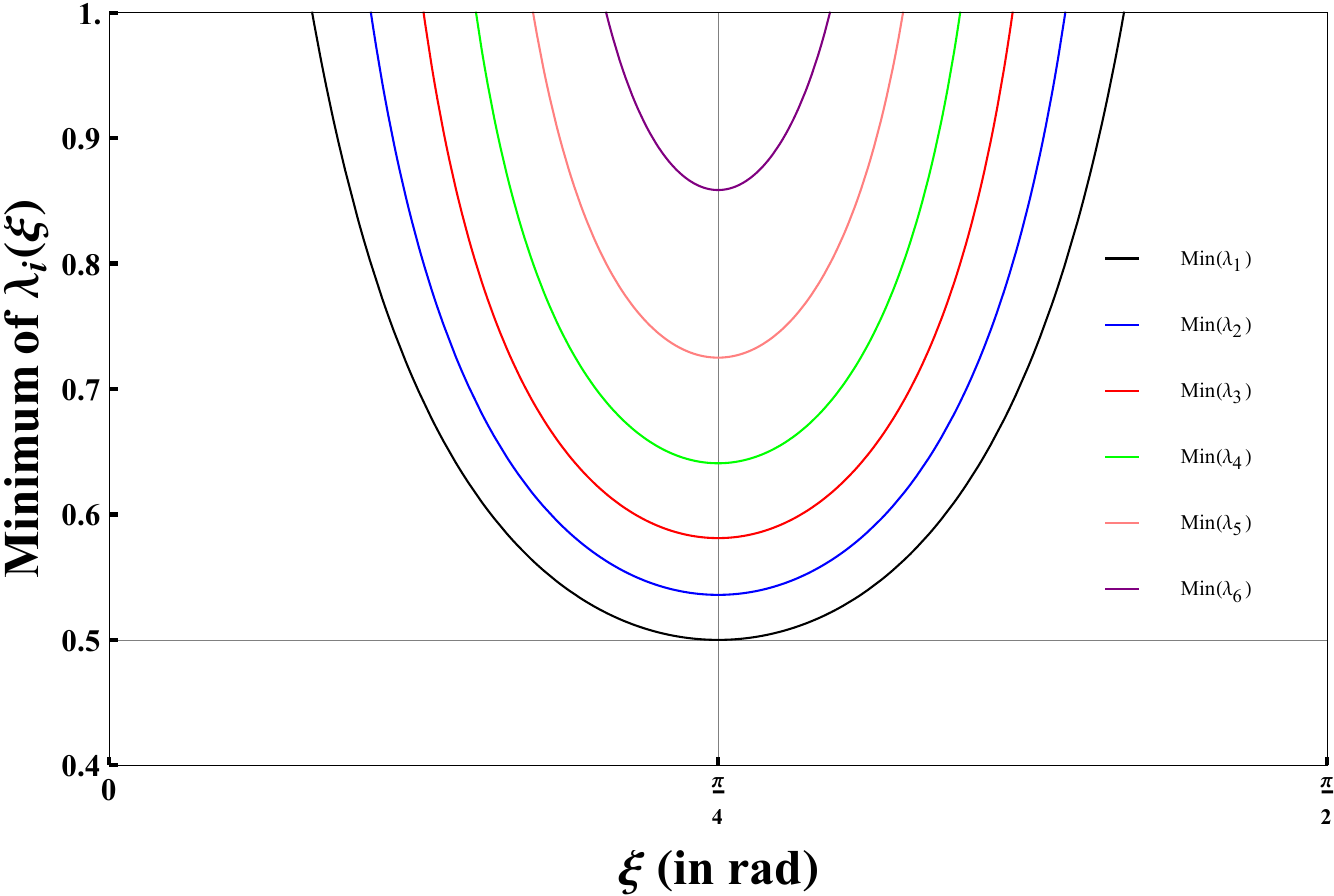}
\caption{\footnotesize (Color Online) $\lbrace(\lambda_i)_{\min}\rbrace_1^6$ are plotted against the state parameter $\xi$ of the initially shared non-maximally entangled state where remote states from equatorial circle ($\theta=\frac{\pi}{2}$) are chosen to be prepared upto Bob$^i$ such that the average RSP-fidelity expressed by Eq.(\ref{nms}) is non-classical in terms of its improvement against our classical strategy.}
\label{rsps1}
\end{figure}

We plot in Fig.\ref{rsps1} the minimum of $\lambda_i$ as a function of the parameter $\xi\in[0,\frac{\pi}{2}]$ of the initially shared non-maximally entangled state for which non-classical average fidelity for RSP is achieved corresponding to the equatorial circle (i.e. $\theta=\frac{\pi}{2}$) of the Bloch sphere. It is observed that $(\lambda_i)_{\min}(\xi)$ is an even function w.r.t. $\xi=\frac{\pi}{4}$ and has minima at $\xi=\frac{\pi}{4}$ for all  Bob$^i$ ($i\leq 6$). It  thus follows that the highest number of Bobs, i.e., $n=6$ can achieve the quantum fidelity of RSP for $\xi=\frac{\pi}{4}$, in the allowed range of $(\lambda_i)_{\min}(\xi) < \lambda_i \leq 1 ~\forall i\leq n$ (see  Table \ref{tab:t1}). As $\xi$ moves from the neighbourhood of $\xi=\frac{\pi}{4}$ towards $\xi=0$ or $\frac{\pi}{2}$ corresponding to the pure product states $\lbrace |01\rangle, |10\rangle \rbrace$, the  bound on the number of Bobs ($n$) sharing RSP with an Alice diminishes gradually from 6 to 0 as the non-classicality is revealed through the violation of our classical strategy. As the entanglement of the initial state decreases, its utility to provide quantum supremacy of the RSP-protocol in terms of $n$ also decreases. Table \ref{tab:n3} shows the bound $n$ corresponding to the region of $\xi$ of the initial state so that $\lambda_{n+1} \nleqslant 1$ occurs by means of the comparison against the classical strategy. Note that, the technique used to calculate the average RSP-fidelity remains same for all the Bobs such that no Bob can prefer a classical or mixed strategy over the quantum strategy at random and the task of sharing RSP stops to progress when the last Bob for a given iteration fails to achieve non-classical advantage through the protocol.

\begin{table}[h!]
\centering
\begin{adjustbox}{width=0.4\textwidth}
 \begin{tabular}{| c | c |} 
 \hline
 $n$ & Range of $\xi$ \\ 
  & (in rad) \\
 \hline
 0 & $[0,\frac{\pi}{12}]\cup[\frac{5\pi}{12},\frac{\pi}{2}]$ \\ 
 \hline
 1 & $(\frac{\pi}{12},0.337]\cup[1.233,\frac{5\pi}{12})$ \\
 \hline
 2 & $(0.337,0.405]\cup[1.165,1.233)$ \\
 \hline
 3 & $(0.405,0.473]\cup[1.098,1.165)$ \\
 \hline
 4 & $(0.473,0.547]\cup[1.024,1.098)$ \\  
 \hline
 5 & $(0.547,0.641]\cup[0.929,1.024)$ \\
 \hline
 6 & $(0.641,0.929)$ \\
 \hline
\end{tabular}
\end{adjustbox}
\caption{At most $n$-number of Bobs can sequentially prepare the remote states from the equatorial circle ($\theta=\frac{\pi}{2}$) of the Bloch sphere such that every Bob upto Bob$^n$ attains the average RSP-fidelity$> \frac{3}{4}$ by making use of Eq.(\ref{nms}) when the parameter $\xi$ of the initially shared non-maximally entangled state remains within the range given by column 2.}
\label{tab:n3}
\end{table}

Note that, the initially shared state between Alice and Bob$^1$ given by Eq.(\ref{nmes}) has concurrence $\mathbb{C}(\rho^1(\xi))=\sin 2\xi$. From Eq.(\ref{nms}), it is evident that as $\mathbb{C}$ increases, the average RSP-fidelity for subsequent Bobs $\forall i$ also increases linearly. The Bell states with $\xi=\frac{\pi}{4}$ achieve the maximum of the average RSP-fidelity given by Eq.(\ref{avfid}). For pure product states ($\xi=0$ or $\frac{\pi}{2}$), the average RSP-fidelity, $f_{av}^{AB^i}=\frac{1}{2} ~\forall i$. The neighbourhood of $\xi=\lbrace 0,\frac{\pi}{2} \rbrace$, i.e., $\xi \in [0,\frac{\pi}{12}] \cup [\frac{5\pi}{12},\frac{\pi}{2}]$ corresponds to $\mathbb{C}(\rho^1(\xi))\leq \frac{1}{2}$ where the initial state is not useful for the implementation of RSP at Alice's side even with a single Bob by comparing the results with our classical strategy.

\begin{figure}[!ht]
\centering
\includegraphics[width=0.75\linewidth]{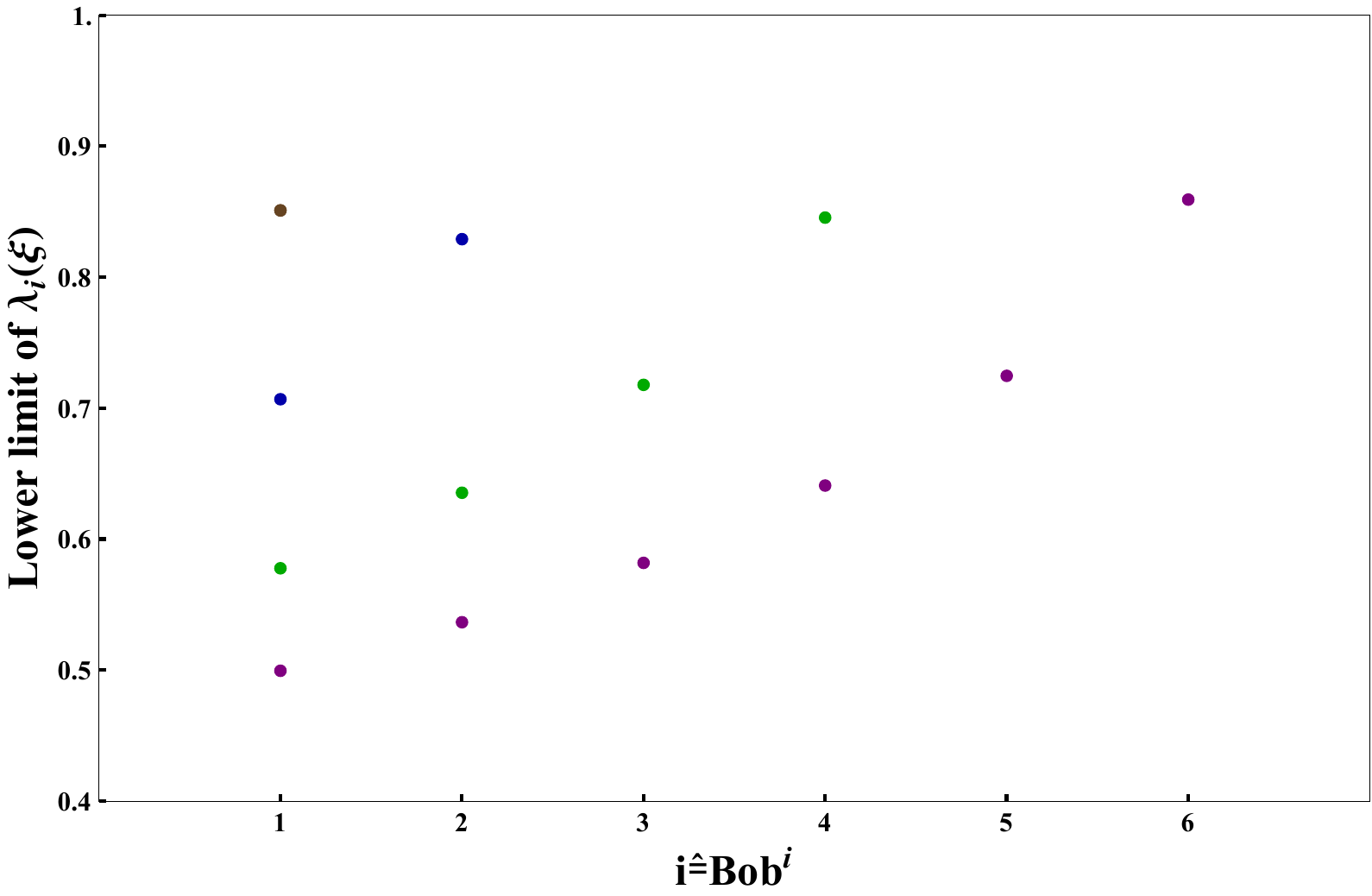}
\caption{\footnotesize (Color Online) The data points corresponding to $(\lambda_i)_{\min}(\xi)$ are plotted w.r.t. the number of Bobs ($i\leq n$) such that, for Bob$^{n+1}$ in the given scenario $(\lambda_{n+1})_{\min}>1$ depending upon the success of the RSP-protocol against our classical strategy. By fixing $\theta=\frac{\pi}{2}$, purple, green, blue and brown data points correspond to the initial states with $\xi=\frac{\pi}{4}, \frac{\pi}{6}, \frac{\pi}{8}$ and $\frac{\pi}{10}$ respectively. The total number of data points for a given $\xi$ implies the optimal bound on the number of Bobs ($n$).}
\label{rsps2}
\end{figure}

We now select a number of initial states with $\xi \in (0,\frac{\pi}{4})$ by taking into account the symmetry of Fig.\ref{rsps2} w.r.t. $\xi=\frac{\pi}{4}$ which yields similar results for the parameters $\xi$ and $\frac{\pi}{2}-\xi$. We show in Fig.\ref{rsps2} that the maximum number of successful Bobs ($n$) sharing RSP with Alice decreases with $\xi_{\neq \frac{\pi}{4}}$ compared to 6 for $\xi=\frac{\pi}{4}$. For example, $n$ is found to be reduced to 4 for $\xi=\frac{\pi}{6}$, to 2 for $\xi=\frac{\pi}{8}$ and to 1 for $\xi=\frac{\pi}{10}$.

\end{proof}

As seen above, pertaining to the equatorial circle, $\theta=\frac{\pi}{2}$, the bound on the number of observers is not improved for non-maximally entangled pure states compared to maximally entangled initial states. Similar results, with suitable post-selection on the outcomes of the subsequent Bobs (giving rise to a RSP-protocol with a success rate less than 50\%), can be shown to hold for the choice of remote qubits from the other circles of the Bloch sphere with polar angles $\theta \neq \frac{\pi}{2}$ ($0< \theta < \pi$) where the transformation $|\psi^i_{\bot}\rangle \rightarrow |\psi^i\rangle$ is forbidden $\forall i$. Therefore, we conjecture, in general, that the bound is less for non-maximally entangled state compared to maximally entangled state for a pre-determined circle of the Bloch sphere. On the other hand, the initial singlet state can be mixed with white noise to form the Werner state given by Eq.(\ref{Werner}) with parameter $c ~(0\leq c \leq 1)$. We show in Appendix \ref{C} that as the mixedness of the initial state increases, the bound on the number of observers sharing RSP decreases than that for the singlet state. No Bob can prepare remote state at Alice's side for the maximally mixed initial state and its neighbouring states.


\section{Conclusions} \label{R4}

The ability to send or prepare a quantum state at a remote location is an important primitive of various protocols in quantum communication and computation. In the task of remote state preparation, an agent prepares a known qubit-state at a remote location with the help of one bit of classical communication and one e-bit of shared entanglement. Various protocols of RSP are possible depending on encoding-decoding strategies or choice of ensemble of states to be prepared. In the present work we consider a specific RSP task in the context of sequential network scenario. Utilizing a single copy of a quantum entangled state  for sharing various kinds of quantum correlations has been previously shown to be effective in different contexts \cite{Silva,Mal,Das,Sasmal,Gupta,Datta}. Applications of the sequential scenario  exist in several directions such as the certification of unbounded randomness~\cite{RNG}. Here we show that multiple observers (Bobs) on one side, who act sequentially and independently of each other, can individually prepare certain states taken from a specified ensemble, at a remote location (Alice's laboratory) with non-classical fidelity. 
 
In order to obtain genuine quantum advantage of a RSP protocol, we first figure out the classical limit of fidelity. It is found that classical fidelity can range from $\frac{3}{4}$ to as high as 1 depending upon the population of states on the chosen circle of the Bloch sphere from where a state is intended to be prepared remotely. For RSP with equatorial input ensemble of states, we found at most six observers can be successful, where all of them act independently. From our results, it follows that for every Bob there is a range of the sharpness parameter for which non-classical advantage persists, and even after the last successful Bob some residual correlation (e.g. non-zero discord) remains which keeps the post-measured state resourceful.

We further propose a new RSP protocol with input states taken from non-equatorial circles and obtain the pertinent classical fidelity. In this case, the number of successful parties reduces from 6 to zero as the chosen circle varies from the neighbourhood of the equatorial circle to the poles of the Bloch sphere. It may be noted that the singlet state is not unique, and the scheme proposed here works equally well for any of the maximally entangled Bell states which are connected by local unitary operations. 
  
Next we consider the performance of non-maximally entangled states in the above mentioned context. Here the bound on the maximum number of observers gets reduced as expected. Even further reduction is shown to ensue when the initial states are mixed with the white noise. As the visibility of the initially shared Werner state reduces, the bound on the number of senders also reduces from 6 to zero gradually, because RSP-fidelity and the mixedness follow an inverse relation with each other~\cite{Discord}. 
 
The framework of remote preparation of a qubit by multiple observers may be extended in various directions in future works. The performance of the protocol under various classical strategies~\cite{Horodecki14}, and utilizing resources such as shared randomness~\cite{BBCM95,Brassard99,Bowles15,Guha21} and weak measurements would be interesting to explore. Further, investigations on collective remote state preparation~\cite{Collective} using multipartite and higher dimensional states are worth of investigation. Besides, consideration of various completely positive trace preserving(CPTP) maps~\cite{DC,Yang} may represent the scenario of RSP more practically. Finally, in future works it will be interesting to explore the quantitative trade-off between the amount of resource accessed and the population utilizing the resource in such multiple observer scenarios. Our present analysis should inspire studies of other RSP protocols, such as RSP with different encoding-decoding operations in the context of the sequential scenario.




\backmatter

\bmhead{Acknowledgments}

The work of SD is financially supported by INSPIRE Fellowship from Department of Science and Technology, Govt. of India (Grant No.C/5576/IFD/2015-16). SM acknowledges support from the Ministry of Science and technology, Taiwan (Grant No. MOST 110- 2124-M-002-012). ASM acknowledges support from the DST Project no. DST/ICPS/QuEST/2018/Q-79. The authors thank the anonymous referees for their useful comments that helped to improve this paper significantly.

\section*{Declarations}

\bmhead{Conflict of interest} The authors have no relevant financial or non-financial interests to disclose.

\bmhead{OpenAccess} This article is licensed under a Creative Commons Attribution 4.0 International License, which permits use, sharing, adaptation, distribution and reproduction in any medium or format, as long as you give appropriate credit to the original author(s) and the source, provide a link to the Creative Commons licence, and indicate if changes were made. The images or other third party material in this article are included in the article's Creative Commons licence, unless indicated otherwise in a credit line to the material. If material is not included in the article's Creative Commons licence and your intended use is not permitted by statutory regulation or exceeds the permitted use, you will need to obtain permission directly from the copyright holder. To view a copy of this licence, visit http://creativecommons.org/licenses/by/4.0/.

\bibliography{sn-bibliography}

\appendix

\begin{appendices}

\section{Average fidelity for preparing remote states from equatorial great circle using singlet state initially}\label{A}

When Bob$^i$ obtains (+) outcome, then the conditional state given by Eq.(\ref{cond1}) is prepared at Alice's side and Alice applies $\sigma_z$ rotation to attain the desired state $|\psi^i\rangle = \frac{1}{\sqrt{2}} \big(|0\rangle + \exp(i\phi_i) |1\rangle\big)$. Hence after rotation, the state becomes $\sigma_z. \rho^i_{A|E_+^{\lambda_i}}. \sigma_z$. So the average RSP-fidelity considering the prepared state and all the desired states from the equatorial great circle of the Bloch sphere becomes,
\begin{equation}
f_{av^+}^{AB^i} = \frac{1}{2\pi} \int_0^{2\pi} \langle \psi^i | \sigma_z. \rho^i_{A|E_+^{\lambda_i}}. \sigma_z | \psi^i \rangle ~d\phi_i, ~~~~~(i\geq 2)
\label{A1}
\end{equation}
with the use of Eq.(\ref{afid}). Now by applying $\sigma_z.|\psi^i_{\bot}\rangle\langle\psi^i_{\bot}|.\sigma_z = |\psi^i\rangle\langle\psi^i|$ for $\theta=\frac{\pi}{2}$ on Eq.(\ref{cond1}), Eq.(\ref{A1}) becomes,
\begin{align}
f_{av^{+}}^{AB^i} =& \Big[ \frac{\lambda_i}{2^{i-1}} ~\prod_{k=1}^{i-1} \Big(1+\sqrt{1-\lambda_k^2}\Big)\Big] ~\frac{1}{2\pi} \int_0^{2\pi} \langle \psi^i | \psi^i \rangle \langle \psi^i | \psi^i \rangle ~d\phi_i \nonumber\\
& + \frac{1}{2} \Big[1 - \frac{\lambda_i}{2^{i-1}} ~\prod_{k=1}^{i-1} \Big(1+\sqrt{1-\lambda_k^2}\Big) \Big] ~\frac{1}{2\pi} \int_0^{2\pi} \langle \psi^i | \sigma_z^2 | \psi^i \rangle ~d\phi_i \nonumber\\
=& \frac{1}{2} + \frac{\lambda_i}{2^i} \prod_{k=1}^{i-1} \Big(1+\sqrt{1-\lambda_k^2}\Big), ~~~~~(i\geq 2).
\end{align}

In a similar fashion, when Bob$^i$ obtains (-) outcome, then the conditional state given by Eq.(\ref{cond2}) is prepared at Alice's side and Alice applies $\mathbb{I}_2$ to obtain the desired state $|\psi^i\rangle$. Therefore with the help of Eq.(\ref{afid}), the average RSP-fidelity becomes,
\begin{align}
f_{av^{-}}^{AB^i} =& \frac{1}{2\pi} \int_0^{2\pi} \langle \psi^i | \mathbb{I}_2. \rho^i_{A|E_-^{\lambda_i}}. \mathbb{I}_2 | \psi^i \rangle ~d\phi_i,  \nonumber\\
=& \Big[ \frac{\lambda_i}{2^{i-1}} ~\prod_{k=1}^{i-1} \Big(1+\sqrt{1-\lambda_k^2}\Big)\Big] ~\frac{1}{2\pi} \int_0^{2\pi} \langle \psi^i | \psi^i \rangle \langle \psi^i | \psi^i \rangle ~d\phi_i \nonumber\\
& + \frac{1}{2} \Big[1 - \frac{\lambda_i}{2^{i-1}} ~\prod_{k=1}^{i-1} \Big(1+\sqrt{1-\lambda_k^2}\Big) \Big] ~\frac{1}{2\pi} \int_0^{2\pi} \langle \psi^i | \psi^i \rangle ~d\phi_i \nonumber\\
=& \frac{1}{2} + \frac{\lambda_i}{2^i} \prod_{k=1}^{i-1} \Big(1+\sqrt{1-\lambda_k^2}\Big), ~~~~~(i\geq 2).
\end{align}

On the other hand, if the desired state at Alice's side is $|\psi^i_{\bot}\rangle = \frac{1}{\sqrt{2}} \big(|0\rangle - \exp(i\phi_i) |1\rangle\big)$, then depending upon Alice's unitaries $\lbrace \mathbb{I}_2, \sigma_z \rbrace$ corresponding to Bob$^i$'s outcomes $\lbrace +,- \rbrace$, the average RSP-fidelity turns out as,
\begin{align}
f_{av^{\pm}}^{AB^i} =& \frac{1}{2\pi} \int_0^{2\pi} \langle \psi^i_{\bot} | \mathbb{I}_2. \rho^i_{A|E_+^{\lambda_i}}. \mathbb{I}_2 | \psi^i_{\bot} \rangle ~d\phi_i,  \nonumber\\
=& \frac{1}{2\pi} \int_0^{2\pi} \langle \psi^i_{\bot} | \sigma_z. \rho^i_{A|E_-^{\lambda_i}}. \sigma_z | \psi^i_{\bot} \rangle ~d\phi_i, \nonumber\\
=& \frac{1}{2} + \frac{\lambda_i}{2^i} \prod_{k=1}^{i-1} \Big(1+\sqrt{1-\lambda_k^2}\Big), ~~~~~(i\geq 2).
\end{align}

Thus according to Eq.(\ref{psi100}) and Eq.(\ref{psi_perp100}), the average RSP-fidelity considered upto Bob$^i$, by taking all the possible measurement outcomes in different contexts into account, can be represented by
\begin{align}
f_{av}^{AB^i} =& \frac{1}{2\pi} \int_0^{2\pi} (p_{+}^i ~\langle \psi^i | \sigma_z. \rho^i_{A|E_+^{\lambda_i}}. \sigma_z | \psi^i \rangle \nonumber\\
&+ p_{-}^i ~\langle \psi^i | \mathbb{I}_2. \rho^i_{A|E_-^{\lambda_i}}. \mathbb{I}_2 | \psi^i \rangle) ~d\phi_i \nonumber\\
=& \frac{1}{2\pi} \int_0^{2\pi} (p_{+}^i ~\langle \psi^i_{\bot} | \mathbb{I}_2. \rho^i_{A|E_+^{\lambda_i}}. \mathbb{I}_2 | \psi^i_{\bot} \rangle \nonumber\\
&+ p_{-}^i ~\langle \psi^i_{\bot} | \sigma_z. \rho^i_{A|E_-^{\lambda_i}}. \sigma_z | \psi^i_{\bot} \rangle) ~d\phi_i \nonumber\\
=& \frac{1}{2} + \frac{\lambda_i}{2^i} \prod_{k=1}^{i-1} \Big(1+\sqrt{1-\lambda_k^2}\Big), ~~~~~(i\geq 2)
\label{A5}
\end{align} 
where the probabilities of finding the conditional states $\rho^i_{A|E_{\pm}^{\lambda_i}}$ (i.e. $p_{\pm}^i$) satisfy $p_{+}^i = p_{-}^i = \frac{1}{2} ~\forall \phi_i,\lambda_i ~(i\geq 1)$. 

\section{Sequential RSP for equatorial states with a different classical strategy}\label{AA}

When the classical strategy is defined under the special implementation of the classical channel where Bob hints Alice about the azimuthal angle of the Bloch sphere without disturbing his own system, then the optimal classical bound under such circumstances comes out to be $\frac{1}{2}+\frac{1}{\pi} \simeq 0.818$ as shown in~\cite{Horodecki14}. By repeating our analysis for the quantum strategy described in the main text,  we find the minimum sharpness of measurements required for quantum advantage for all subsequent Bobs for the performance of RSP in the multiple observer scenario as follows.

\begin{table}[h!]
\centering
\begin{adjustbox}{width=0.4\textwidth}
 \begin{tabular}{| c | c |} 
 \hline
 $i ~\hat{=}$ Bob$^i$ & Range of $\lambda_i$ \\
 &  $\hat{=}$ 2nd classical strategy \\ 
 \hline
 1 & (0.637 - 1] \\ 
 \hline
 2 & (0.719 - 1] \\
 \hline
 3 & (0.848 - 1]\\
 \hline
\end{tabular}
\end{adjustbox}
\caption{The domain of sharpness parameters for which 3 Bobs altogether attain the average RSP-fidelity given by Eq.(\ref{avfid}) with a gain compared to the classical bound of $\frac{1}{2}+\frac{1}{\pi}$.}
\label{tab:new}
\end{table}

Here we observe that 3 Bobs consecutively achieve success in the 100\% successful RSP-protocol while Bob$^4$ achieves a maximum of 0.79 subject to the sharp measurement done by him. However, it is not obvious how to  generalise the above classical strategy for states corresponding to the non-equatorial circles of the Bloch sphere.

\section{Sequential RSP with $\theta=\tan^{-1} \sqrt{2}, \xi=\frac{\pi}{4}$}\label{B}

\begin{result}
There can be at most 3 Bobs who sequentially prepare the remote states in Alice's lab with success achieved by all of them against our optimal classical strategy by choosing from a non-equatorial circle with polar angle ($\theta$) either $\tan^{-1} \sqrt{2}$ or $(\pi - \tan^{-1} \sqrt{2})$ (in radian) on the Bloch sphere when Alice and Bob$^1$ share a maximally entangled state ($\xi=\frac{\pi}{4}$) initially.
\end{result}

\begin{proof}
The average fidelity of preparing remote state at Alice's side by Bob$^i$ under suitable post-selection method at the end of the i-th step of performing the task of RSP with 50\% success as $\theta\in (0,\pi)$, is given by,
\begin{equation}
f_{av}^{AB^i} = \frac{f_{cl}^{\max}|_{(\theta=\tan^{-1}\sqrt{2})}}{2} + \frac{1}{4} \Big[1 + \frac{\lambda_i}{3^{i-1}} \prod_{k=1}^{i-1} \Big(1+2\sqrt{1-\lambda_k^2}\Big)\Big], ~~(i\geq 2)
\end{equation}
with the use of Eq.(\ref{psi50}), where $f_{cl}^{\max}|_{(\theta=\tan^{-1}\sqrt{2})}=f_{cl}^{\max}|_{(\theta=\pi-\tan^{-1}\sqrt{2})}=0.803$.

There exists a finite range of $\lambda_i$ in each step upto Bob$^i$ ($i \geq 2$), for which non-classical advantage through RSP-fidelity can be gained i.e. $f_{av}^{AB^i}> 0.803$ in comparison with our classical strategy. It is in accordance with the possible range of sharpness parameters $\lambda_j$, ($j<i$) employed by  all the previous Bobs upto Bob$^j$ $\forall j \in \lbrace 1,2,...,i-1 \rbrace$ to achieve RSP at Alice's side with quantum average fidelity. This is demonstrated in Table \ref{tab:t2}.

\begin{table}[h!]
\centering
\begin{adjustbox}{width=0.3\textwidth}
 \begin{tabular}{| c | c |} 
 \hline
 $i ~\hat{=}$ Bob$^i$ & Range of $\lambda_i$ \\ 
 \hline
 1 & (0.605 - 1] \\ 
 \hline
 2 & (0.701 - 1] \\
 \hline
 3 & (0.866 - 1] \\
 \hline
\end{tabular}
\end{adjustbox}
\caption{The range of sharpness parameters $\lambda_i\in((\lambda_i)_{\min},1]$ for which 3 Bobs can independently share the task of RSP at Alice's side with average fidelity$f_{av}^{AB^i}> 0.803$ as compared to the classical strategy.}
\label{tab:t2}
\end{table}

Corresponding to the 4-th Bob, there exists no region of $\lambda_{4} \in [0,1]$ for which Bob$^{4}$ can be successful in preparing remote state at Alice's side whenever all the Bobs choose a circle with polar angle either $\tan^{-1} \sqrt{2}$ or $(\pi - \tan^{-1} \sqrt{2})$ from the Bloch sphere. Bob$^{4}$ achieves optimum average RSP-fidelity, i.e. 0.733 even by doing a projective measurement. Therefore, at most 3 independent Bobs are able to exhibit the half successful task of RSP in the given scenario where all the Bobs apply similar kind of approach sequentially to achieve success against the optimal classical strategy.
\end{proof}

\section{Sequential RSP with Werner state}\label{C}

\begin{result}
The maximum number of Bobs preparing remote states sequentially at Alice's side with success accomplished by all of them uniformly compared to our optimal classical strategy is upper bounded by 6 when Alice shares the Werner state with Bob$^1$ initially. As the mixedness of the initial state increases, the maximum number of Bobs reduces to zero gradually.
\end{result}

\begin{proof}
We suppose that, $\rho^1 = \rho_W$ given by Eq.(\ref{Werner}) which has mixedness quantified by the linear entropy~\cite{Peters} as $S_L(\rho_W)= \frac{4}{3}(1-\operatorname{Tr}[\rho_W^2]) = 1-c^2$ where the Werner parameter $c\in [0,1]$. There are a number of Bobs at one half of the Werner state who wants to sequentially prepare remote states from the equatorial circle of the Bloch sphere ($\theta=\frac{\pi}{2}$) to Alice present at the other half of the shared state. For the RSP protocol in this scenario, the average RSP-fidelity corresponding to Bob$^1$ by applying Eq.(\ref{psi100}) or Eq.(\ref{psi_perp100}) becomes 
\begin{equation}
f_{av}^{AB^1} = \frac{1+c \lambda_1}{2} = \frac{1 + \sqrt{1-S_L(\rho_W)} \lambda_1}{2}
\end{equation}
which can not produce 100\% success of the protocol (i.e. $f_{av}^{AB^1} \neq 1$) even for the sharp measurement at Bob$^1$'s side (i.e. $\lambda_1 =1$) until the initial state is the singlet state with $c=1$ given by Eq.(\ref{singlet}). Here $f_{av}^{AB^1}$ is higher than the classical fidelity bound i.e. $\frac{3}{4}$ when $\frac{1}{2c} < \lambda_1 \leq 1$. Hence $(\lambda_1)_{\min}=\frac{1}{2\sqrt{1-S_L(\rho_W)}}$ which increases with the mixedness and is the lowest when $c=1$ i.e. for the singlet state. When $\rho^1$ is the maximally mixed state ($c=0$), then $f_{av}^{AB^1} = \frac{1}{2}$ and is not capable of doing RSP with quantum advantage even by employing a single Bob. And the allowed range of $(\lambda_1)_{\min} < 1$ is analogous to $c > \frac{1}{2}$. Therefore no Bob can prepare remote states from the equatorial plane of the Bloch sphere at Alice's side when $0 \leq c \leq \frac{1}{2}$.

For subsequent Bobs ($i\geq 2$), the average RSP-fidelity as per Eq.(\ref{psi100}) or Eq.(\ref{psi_perp100}), by using $p_{+}^i=p_{-}^i = \frac{1}{2} ~\forall c,\phi_i,\lambda_i$, becomes
\begin{equation}
f_{av}^{AB^i} = \frac{1}{2} + \frac{c \lambda_i}{2^i} \prod_{k=1}^{i-1} \Big(1+\sqrt{1-\lambda_k^2}\Big), ~~~~~(i\geq 2)
\label{avfid1}
\end{equation}
which is larger than $\frac{3}{4}$ corresponding to the equatorial circle ($\theta=\frac{\pi}{2}$) of the Bloch sphere when $(\lambda_i)_{\min} < \lambda_i \leq 1$. For instance, $(\lambda_2)_{\min} = 4c-2\sqrt{4c^2 -1}$ by comparing the average RSP-fidelity with its classical counterpart. We plot in Fig.\ref{rspsn2} $(\lambda_i)_{\min} ~\forall i$ within the allowed region of $\lambda_i ~\forall i$ and show that, there can not be more than 6 sequential Bobs who can successfully share RSP with Alice in comparison with our optimal classical strategy for the initially shared Werner state in the given framework. It can be observed from Fig.\ref{rspsn2} that, $(\lambda_i)_{\min} ~\forall i\leq 6$ is inversely proportional to $c$. Thus as the mixedness of the initial state increases, $(\lambda_i)_{\min} ~\forall i$ becomes larger and it becomes difficult for more number of Bobs to gain success in the RSP-protocol and the optimal bound on the number of Bobs ($n$) decreases.

\begin{figure}[!ht]
\centering
\includegraphics[width=0.75\linewidth]{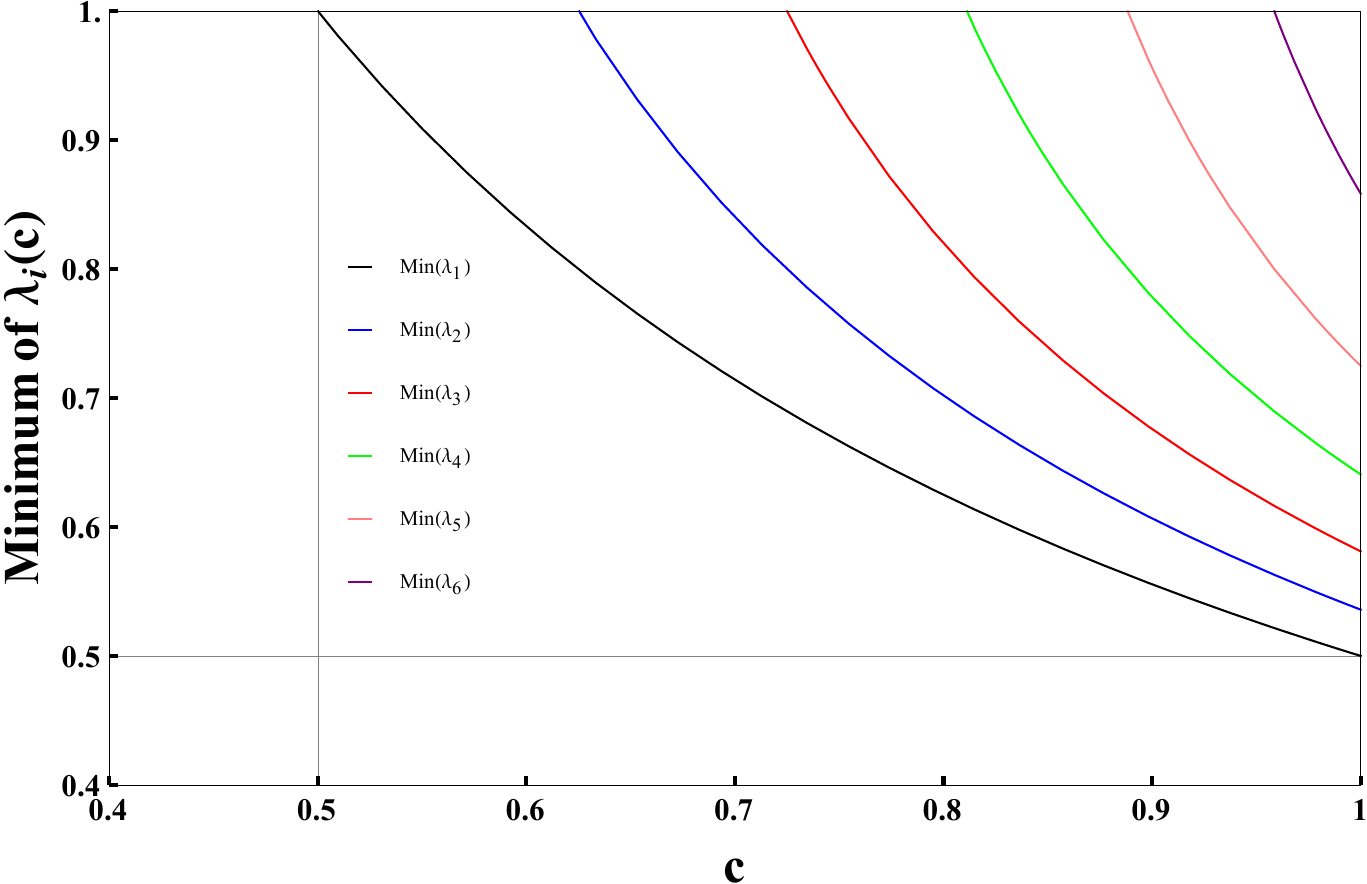}
\caption{\footnotesize (Color Online) $\lbrace(\lambda_i)_{\min}\rbrace_1^6$ are plotted against the Werner state parameter $c$ of the initial state where the remote states are prepared from the equatorial circle ($\theta=\frac{\pi}{2}$) in order to achieve non-classical average RSP-fidelity by applying Eq.(\ref{avfid1}) corresponding to Bob$^i$.}
\label{rspsn2}
\end{figure}

The bound, $n$ is maximum, i.e. 6 for the pure initial (singlet) state with $c=1$, whereas it gradually decreases depending upon the mixedness of the initial state and becomes zero for $c\leq \frac{1}{2}$ or the linear entropy $S_L(\rho_W) \geq \frac{3}{4}$. In Table \ref{tab:n4}, we demonstrate the bound on the successful Bobs ($n$) corresponding to the regions of Werner parameters of $\rho_W$ for which all the $n$-number of Bobs achieve non-classical advantage. We assume that all the Bobs follow the same strategy to compute the average RSP-fidelity while the task of sharing RSP truncates subject to the failure of the last Bob for a given iteration to attain the non-classical supremacy.

\begin{table}[h!]
\centering
\begin{adjustbox}{width=0.25\textwidth}
 \begin{tabular}{| c | c |} 
 \hline
 $n$ & Range of $c$ \\ 
 \hline
 0 & $[0,\frac{1}{2}]$ \\ 
 \hline
 1 & $(\frac{1}{2},0.625]$ \\
 \hline
 2 & $(0.625,0.725]$ \\
 \hline
 3 & $(0.725,0.811]$ \\
 \hline
 4 & $(0.811,0.888]$ \\  
 \hline
 5 & $(0.888,0.959]$ \\
 \hline
 6 & $(0.959,1]$ \\
 \hline
\end{tabular}
\end{adjustbox}
\caption{At most $n$-number of Bobs can sequentially prepare the remote states from equatorial circle ($\theta=\frac{\pi}{2}$) of the Bloch sphere such that every Bob upto Bob$^n$ are capable for doing RSP with average RSP-fidelity$> \frac{3}{4}$ by employing Eq.(\ref{avfid1}) when the visibility of the Werner state lies within the region given by column 2.}
\label{tab:n4}
\end{table}

If the remote states are chosen from the non-equatorial circles ($\theta \neq \frac{\pi}{2}$), then corresponding to Bob$^1$ the average RSP-fidelity under suitable post-selection technique by employing Eq.(\ref{psi50}) or Eq.(\ref{psi_perp50}) becomes
\begin{equation}
f_{av}^{AB^1} = \frac{f_{cl}^{\max}}{2} + \frac{1+c \lambda_1}{4}
\end{equation}
and $f_{av}^{AB^1} > f_{cl}^{\max}$ occurs when $\lambda_1 > \frac{1+\cos 2\theta + \sin^3 \theta}{2c}$ which increases with the mixedness and becomes the smallest for the singlet state. No Bob is able to successfully prepare remote states from the non-equatorial circles of the Bloch sphere at Alice's side when the initial state is mixed maximally. On the other hand, corresponding to subsequent Bobs ($i\geq 2$), the average fidelity of preparing remote states from the non-equatorial circles of the Bloch sphere ($\theta \neq \frac{\pi}{2}$, $\theta \in (0,\pi)$) under the method of post-selection, by applying $p_{+}^i=p_{-}^i = \frac{1}{2} ~\forall c,\theta,\phi_i,\lambda_i$, becomes
\begin{align}
f_{av}^{AB^i} =& \frac{f_{cl}^{\max}}{2} + \frac{1}{4} + \frac{c \lambda_i}{4} \Big[ \cos^2\theta ~\prod_{k=1}^{i-1} \Big( \cos^2\theta + \sin^2\theta \sqrt{1-\lambda_k^2} \Big) \nonumber\\
&+ \frac{\sin^2\theta}{2^{i-1}} ~\prod_{k=1}^{i-1} \Big( \sin^2\theta + (\cos^2\theta +1) \sqrt{1-\lambda_k^2} \Big) \Big], ~~~(i\geq 2)
\end{align}
according to Eq.(\ref{psi50}) or Eq.(\ref{psi_perp50}) which decreases as the mixedness of the initial state increases, therefore it becomes higher than the classical fidelity bound for the lesser number of eligible Bobs in the sequence from the analogy of the previous discussion. The action of each sequential Bob remains uniform throughout the sharing of RSP-protocol which continues till the last Bob in a given sequence achieves the non-classical advantage by comparing with the optimal classical strategy proposed by us. Note that, RSP-protocol for the non-equatorial circles of the Bloch sphere can achieve the success of less than 50\% in case of Werner state as the initial state. As a consequence, we generally summarize that, no more than 6 Bobs can share the task of RSP with an Alice when the Werner state is shared at the beginning and the optimum number of Bobs changes inversely with the mixedness of the Werner state.
\end{proof}

\end{appendices}

\end{document}